\documentclass[12pt,letterpaper]{article}
\usepackage{epsfig,rotating,setspace,latexsym,amsmath,epsf,amssymb,bm}
\usepackage{cite}

\title{Outer Bounds for Multiple Access Channels with Feedback using Dependence Balance\thanks{This work was supported by
NSF Grants CCF $04$-$47613$, CCF $05$-$14846$, CNS $07$-$16311$ and
CCF 07-29127. }}

\author{Ravi Tandon \qquad Sennur Ulukus \\
\normalsize Department of Electrical and Computer Engineering\\
\normalsize University of Maryland, College Park, MD 20742 \\
\normalsize {\it ravit@umd.edu} \qquad {\it ulukus@umd.edu}}

\newtheorem{Lem}{Lemma}

\newenvironment{proof}[1]{\medskip\par\noindent
{\bf Proof:\,}\,#1}{{\mbox{\,$\blacksquare$}\par}}
\begin{document}
\maketitle

\begin{abstract}
We use the idea of dependence balance \cite{Hekstra_Willems:1989}
to obtain a new outer bound for the capacity region of the
discrete memoryless multiple access channel with noiseless
feedback (MAC-FB). We consider a binary additive noisy MAC-FB
whose feedback capacity is not known. The binary additive noisy
MAC considered in this paper can be viewed as the discrete
counterpart of the Gaussian MAC-FB. Ozarow \cite{OZ:1984}
established that the capacity region of the two-user Gaussian
MAC-FB is given by the cut-set bound. Our result shows that for
the discrete version of the channel considered by Ozarow, this is
not the case. Direct evaluation of our outer bound is intractable
due to an involved auxiliary random variable whose large
cardinality prohibits an exhaustive search. We overcome this
difficulty by using functional analysis to explicitly evaluate our
outer bound. Our outer bound is strictly less than the cut-set
bound at all points on the capacity region where feedback
increases capacity. In addition, we explicitly evaluate the
Cover-Leung achievable rate region \cite{CL:1981} for the binary
additive noisy MAC-FB in consideration. Furthermore, using the
tools developed for the evaluation of our outer bound, we also
explicitly characterize the boundary of the feedback capacity
region of the binary erasure MAC, for which the Cover-Leung
achievable rate region is known to be tight. This last result
confirms that the feedback strategies developed in
\cite{KramerMAC:1999} for the binary erasure MAC are capacity
achieving.
\end{abstract}

\newpage

\section{Introduction}
Noiseless feedback can increase the capacity region of the
discrete memoryless MAC, unlike for the single-user discrete
memoryless channel. This was shown by Gaarder and Wolf in
\cite{GW:1975} for the binary erasure MAC, which is defined as
$Y=X_{1}+X_{2}$. Ozarow showed in \cite{OZ:1984} that feedback can
also increase the capacity region of a two-user Gaussian MAC-FB. A
constructive achievability scheme based on the classical
Kailath-Schalkwijk \cite{SK:1966} feedback scheme was shown to be
optimal for the two-user Gaussian MAC-FB. Moreover, the cut-set
outer bound was shown to be tight in this case.

Subsequently, Cover and Leung obtained an achievable rate region
for the general MAC-FB based on block Markov superposition coding
\cite{CL:1981}. Even though this region is in general larger than
the capacity region of the MAC without feedback, it is not optimal
for the two-user Gaussian MAC-FB, as was shown in \cite{OZ:1984}.
Kramer \cite{Kramer:thesis} used the notion of directed
information to obtain an expression for the capacity region of the
discrete memoryless MAC-FB. Unfortunately, this expression is in
an incomputable non-single-letter form. Recently, Bross and
Lapidoth \cite{BrossLapidoth:2005} proposed an achievable rate
region for the two-user discrete memoryless MAC-FB and showed that
their region includes the Cover-Leung region, the inclusion being
strict for some channels.

For a specific class of MAC-FB, Willems \cite{WL:1982} developed
an outer bound that equals the Cover-Leung achievable rate region.
For this class of MAC-FB, each channel input (say $X_{1}$) should
be expressible as a deterministic function of the other channel
input ($X_{2}$) and the channel output ($Y$). The binary erasure
MAC considered by Gaarder and Wolf, where $Y=X_{1}+X_{2}$, falls
into this class of channels. Therefore, Cover-Leung region is the
feedback capacity region for the binary erasure MAC.

A general outer bound for MAC-FB is the cut-set bound. Although
the cut-set bound was shown to be tight for the two-user Gaussian
MAC-FB, it is in general loose. An intuitive reason for the
cut-set bound to be loose for the general MAC-FB is its
permissibility of arbitrary input distributions, some of which
yielding rates which may not be achievable. For instance, even
though Cover-Leung achievability scheme introduces correlation
between $X_{1}$ and $X_{2}$, it is a limited form of correlation,
as the channel inputs are conditionally independent given an
auxiliary random variable, whereas the cut-set bound allows all
possible correlations.

The idea of dependence balance was introduced by Hekstra and
Willems in \cite{Hekstra_Willems:1989} to obtain an outer bound on
the capacity region of the single-output two-way channel. The
basic idea behind this outer bound is to restrict the set of
allowable input distributions, consequently restricting arbitrary
correlation between channel inputs. The authors also developed a
parallel channel extension for the dependence balance bound. The
parallel channel extension can be interpreted as follows: the
parallel channel output can be considered as a genie aided
information which is made available at both transmitters and the
receiver and it also effects the set of allowable input
distributions through the dependence balance bound. Depending on
the choice of the genie information (which is equivalent to
choosing a parallel channel), there is an inherent tradeoff
between the set of allowable input distributions and the excessive
mutual information rate terms which appear in the rate expressions
as a consequence of the parallel channel output. We will exploit
this tradeoff provided by the parallel channel extension of the
dependence balance bound to obtain a strict improvement over the
cut-set bound for a particular MAC whose feedback capacity is not
known.

To motivate the choice of our MAC, consider the binary erasure MAC
used by Gaarder and Wolf given by $Y=X_{1}+X_{2}$. If we introduce
binary additive noise at the channel output, then the channel
becomes $Y=X_{1}+X_{2}+N$, where all $X_{1}$, $X_{2}$ and $N$ are
binary and $N$ has a uniform distribution. This is a
non-deterministic noisy MAC which does not fall into any class of
channels for which the feedback capacity is known. We should mention
that this particular MAC was extensively studied by Kramer in
\cite{Kramer:thesis, Kramer:2003}, where the first improvement over
the Cover-Leung achievable rate region was obtained.

We extend the idea of dependence balance to obtain an outer bound
for the entire capacity region of this binary additive noisy
MAC-FB. Direct evaluation of the parallel channel based dependence
balance bound is intractable due to an involved auxiliary random
variable whose large cardinality prohibits an exhaustive search.
We use composite functions and their properties to obtain a simple
characterization for our bound. Our outer bound strictly improves
upon the cut-set bound at all points on the boundary where
feedback increases capacity. In addition, we explicitly evaluate
the Cover-Leung achievable rate region for our binary additive
noisy MAC-FB.

We particularly focus on the symmetric-rate\footnote{By
symmetric-rate point, we refer to the maximum rate $R$ such that the
rate pair $(R,R)$ lies in the capacity region of MAC-FB.} point on
the feedback capacity region of this channel. Cover-Leung's
achievable symmetric-rate for this channel was obtained in
\cite{Kramer:2003} as $0.43621$ bits/transmission. In
\cite{Kramer:2003}, Kramer obtained an improved symmetric-rate inner
bound as $0.43879$ bits/transmission by using superposition coding
and binning with code trees. The cut-set upper bound on the
symmetric-rate was obtained in \cite{Kramer:2003} as $0.45915$
bits/transmission. We obtain a symmetric-rate upper bound of
$0.45330$ bits/transmission which strictly improves upon the cut-set
bound. Furthermore, we also show that a binary and uniform selection
of the involved auxiliary random variable is sufficient to obtain
our symmetric-rate upper bound.

It should be remarked that the channel we consider in this paper
can be thought of as the discrete counterpart of the channel
considered by Ozarow \cite{OZ:1984}. Although the cut-set bound
was shown to be tight for the two-user Gaussian MAC-FB, our result
shows that the cut-set bound is not tight for the discrete version
of the additive noisy MAC-FB.

As an application of the properties of the composite functions
developed in this paper, we are able to obtain the entire boundary
of the capacity region of the binary erasure MAC-FB. The
evaluation of the asymmetric rate pairs on the boundary of the
feedback capacity region of the binary erasure MAC was mentioned
as an open problem in \cite{VinckPost:1985}. It was shown in
\cite{WillemsonFB:1984} that a binary and uniform auxiliary random
variable $T$ is sufficient to attain the sum-rate point on the
capacity region of the binary erasure MAC-FB.  We show here that
this is also the case for any asymmetric rate point on the
boundary of the feedback capacity region. This result also
complements the work of Kramer \cite{KramerMAC:1999}, where
feedback strategies were developed for the binary erasure MAC-FB
and it was shown that these strategies achieve all rates yielded
by a binary selection of the auxiliary random variable $T$ in the
capacity region. Our result hence shows in effect that the
feedback strategies developed in \cite{KramerMAC:1999} for binary
erasure MAC are optimal and capacity achieving.

\section{System Model}
A discrete memoryless two-user MAC-FB (see Figure 1) is defined by
the following: two input alphabets $\mathcal{X}_{1}$ and
$\mathcal{X}_{2}$, an output alphabet $\mathcal{Y}$, and the channel
defined by a probability transition function $p(y|x_{1},x_{2})$ for
all $(x_{1},x_{2},y)\in\mathcal{X}_{1}\times\mathcal{X}_{2}\times
\mathcal{Y}$. A $(n,M_{1},M_{2},P_{e})$ code for the MAC-FB consists
of two sets of encoding functions $f_{1i},f_{2i}$ for $i=1,\ldots,n$
and a decoding function $g$
\begin{align}
f_{1i}&:\mathcal{M}_{1}\times\mathcal{Y}^{i-1}\rightarrow
\mathcal{X}_{1}, \quad i=1,\ldots,n\nonumber\\
f_{2i}&:\mathcal{M}_{2}\times\mathcal{Y}^{i-1}\rightarrow
\mathcal{X}_{2}, \quad i=1,\ldots,n\nonumber\\
g&: \mathcal{Y}^{n} \rightarrow \mathcal{M}_{1} \times
\mathcal{M}_{2}\nonumber
\end{align}
The two transmitters produce independent and uniformly distributed
messages $W_{1} \in \{1,\ldots,\break M_{1}\}$ and $W_{2} \in
\{1,\ldots,M_{2}\}$, respectively, and transmit them through $n$
channel uses. The average error probability is defined as
$P_{e}=Pr(g(Y^{n})\neq (W_{1},W_{2}))$. A rate pair
$(R_{1},R_{2})$ is said to be achievable for MAC-FB if for any
$\epsilon \geq 0$, there exists a pair of  $n$ encoding functions
$\{f_{1i}\}_{i=1}^{n}$, $\{f_{2i}\}_{i=1}^{n}$, and a decoding
function $g$ such that $R_{1}\leq \text{log}(M_{1})/n$, $R_{2}\leq
\text{log}(M_{2})/n$ and $P_{e}\leq \epsilon$ for sufficiently
large $n$. The capacity region of MAC-FB is the closure of the set
of all achievable rate pairs $(R_{1},R_{2})$.

\section{Cut-Set Outer Bound for MAC-FB}
By applying Theorem 14.10.1 in \cite{Cover:book}, the cut-set
outer bound on the capacity region of MAC-FB can be obtained as:
\begin{align}
\mathcal{CS} = \Big\{(R_{1},R_{2}):\hspace{0.05in} &R_{1}\leq I(X_{1};Y|X_{2})\label{CS1}\\
&R_{2}\leq I(X_{2};Y|X_{1})\label{CS2}\\
&R_{1}+R_{2}\leq I(X_{1},X_{2};Y)\label{CS3}\Big\}
\end{align}
where the random variables $(X_{1},X_{2},Y)$ have the joint
distribution
\begin{align}
p(x_{1},x_{2},y)=p(x_{1},x_{2})p(y|x_{1},x_{2})\label{CS4}
\end{align}
The cut-set outer bound allows all input distributions
$p(x_{1},x_{2})$, which makes it seemingly loose since an
achievable scheme might not achieve arbitrary correlation and
rates given by the cut-set bound. Our aim is to restrict the set
of allowable input distributions by using a dependence balance
approach.

\begin{figure}[t]
  \centering
  \centerline{\epsfig{figure=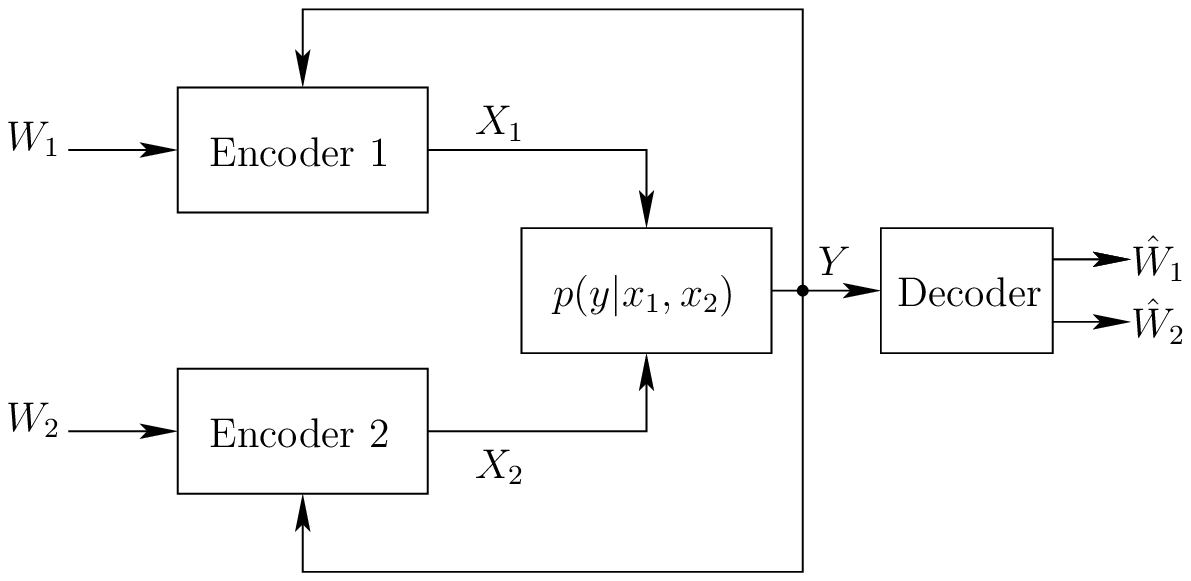,width=9cm}}
  { Figure 1: The multiple access channel with noiseless feedback (MAC-FB).}\medskip
\end{figure}
\section{Dependence Balance Outer Bound for
MAC-FB} Hekstra and Willems \cite{Hekstra_Willems:1989} showed that
the capacity region of MAC-FB is contained within $\mathcal{DB}$,
where
\begin{align}
\mathcal{DB}= \Big\{(R_{1},R_{2}):\hspace{0.05in} &R_{1}\leq I(X_{1};Y|X_{2},T)\label{DBa1}\\
&R_{2}\leq I(X_{2};Y|X_{1},T)\label{DBa2}\\
&R_{1}+R_{2}\leq I(X_{1},X_{2};Y)\Big\}\label{DBa3}
\end{align}
where the random variables $(X_{1},X_{2},Y,T)$ have the joint
distribution
\begin{align}
p(t,x_{1},x_{2},y)=p(t)p(x_{1},x_{2}|t)p(y|x_{1},x_{2})\label{DBa4}
\end{align}
and also satisfy the following dependence balance bound
\begin{align}
I(X_{1};X_{2}|T)\leq I(X_{1};X_{2}|Y,T)\label{DBa5}
\end{align}
where $T$ is subject to a cardinality constraint of
$|\mathcal{T}|\leq |\mathcal{X}_{1}||\mathcal{X}_{2}|+2$. The
dependence balance bound restricts the set of input distributions
in the sense that it allows only those input distributions
$p(t,x_{1},x_{2})$ which satisfy (\ref{DBa5}). It should be noted
that by ignoring the constraint in (\ref{DBa5}), one obtains the
cut-set bound.

\section{Adaptive Parallel Channel Extension of the Dependence Balance Bound}
In \cite{Hekstra_Willems:1989}, Hekstra and Willems also developed
an adaptive parallel channel extension for the dependence balance
bound which is given as follows: Let $\Delta(\mathcal{U})$ denote
the set of all distributions of $U$ and
$\Delta(\mathcal{U}|\mathcal{V})$ denote the set of all conditional
distributions of $U$ given $V$. Then for any mapping
$F:\Delta(\mathcal{X}_{1}\times\mathcal{X}_{2})\rightarrow
\Delta(\mathcal{Z}|\mathcal{X}_{1}\times\mathcal{X}_{2}\times\mathcal{Y})$,
the capacity region of the MAC-FB is contained in
\begin{align}
\mathcal{DB}_{PC}= \Big\{(R_{1},R_{2}):\hspace{0.05in}&R_{1}\leq I(X_{1};Y,Z|X_{2},T)\label{DBPC1}\\
&R_{2}\leq I(X_{2};Y,Z|X_{1},T)\label{DBPC2}\\
&R_{1}\leq I(X_{1};Y|X_{2})\label{DBPC3}\\
&R_{2}\leq I(X_{2};Y|X_{1})\label{DBPC4}\\
&R_{1}+R_{2}\leq I(X_{1},X_{2};Y)\label{DBPC5}\\
&R_{1}+R_{2}\leq I(X_{1},X_{2};Y,Z|T) \Big\}\label{DBPC6}
\end{align}
where the random variables $(X_{1},X_{2},Y,Z,T)$ have the joint
distribution
\begin{align}
&p(t,x_{1},x_{2},y,z)=p(t)p(x_{1},x_{2}|t)p(y|x_{1}x_{2})p^{+}(z|x_{1},x_{2},y,t)\label{DBPC7}
\end{align}
such that for all $t$
\begin{align}
&p^{+}(z|x_{1},x_{2},y,t)=F(p_{X_{1}X_{2}}(x_{1},x_{2}|t))\label{DBPC8}
\end{align}
and such that
\begin{align}
&I(X_{1};X_{2}|T)\leq I(X_{1};X_{2}|Y,Z,T)\label{DBPC9}
\end{align}
where $T$ is subject to a cardinality bound of $|\mathcal{T}|\leq
|\mathcal{X}_{1}||\mathcal{X}_{2}|+3$.

We should remark that the parallel channel (defined by
$p^{+}(z|x_{1},x_{2},y,t)$) is selected apriori, and for every
choice of the parallel channel, one obtains an outer bound on the
capacity region of MAC-FB, which is in general tighter than the
cut-set bound. The set of allowable input distributions
$p(t,x_{1},x_{2})$ are those which satisfy the constraint in
(\ref{DBPC9}). Also note that only the right hand side of
(\ref{DBPC9}), i.e., only $I(X_{1};X_{2}|Y,Z,T)$, depends on the
choice of the parallel channel. By carefully selecting
$p^{+}(z|x_{1},x_{2},y,t)$, one can reduce $I(X_{1};X_{2}|Y,Z,T)$,
thereby making the constraint in (\ref{DBPC9}) more stringent,
consequently reducing the set of allowable input distributions. To
obtain an improvement over the cut-set bound, we need to select a
``good" parallel channel such that it restricts the input
distributions to a small allowable set and yields small values of
$I(X_{1};Z|Y,X_{2},T)$ and $I(X_{2};Z|Y,X_{1},T)$ at the same
time. These two mutual information ``leak" terms are the extra
terms that appear in (\ref{DBPC1}) and (\ref{DBPC2}) relative to
the rates appearing in (\ref{DBa1}) and (\ref{DBa2}),
respectively.

To motivate the choice of our particular parallel channel, first
consider a trivial choice of $Z$: $Z=\phi$ (a constant). For this
choice of $Z$, (\ref{DBPC9}) reduces to (\ref{DBa5}) and we are
not restricting the set of allowable input distributions any more
than the $\mathcal{DB}$ bound. Moreover, for a constant selection
of $Z$, (\ref{DBPC1}) and (\ref{DBPC2}) reduce to (\ref{DBa1}) and
(\ref{DBa2}), respectively. Thus, a constant selection of $Z$ for
$\mathcal{DB}_{PC}$ is equivalent to $\mathcal{DB}$ itself.

Also note that the smallest value of $I(X_{1};X_{2}|Y,Z,T)$ is zero.
Thus, it follows that if we select a parallel channel such that
$I(X_{1};X_{2}|Y,Z,T)=0$ for every input distribution
$p(t,x_{1},x_{2})$, then $I(X_{1};X_{2}|T)=0$ by (\ref{DBPC9}).
Hence, the smallest set of input distributions permissable by
$\mathcal{DB}_{PC}$ consists of those $p(t,x_{1},x_{2})$ for which
$X_{1}$ and $X_{2}$ are conditionally independent given $T$.
Furthermore, for a parallel channel such that
$I(X_{1};X_{2}|Y,Z,T)=0$, the bound in (\ref{DBPC6}) is redundant.
This can be seen from:
\begin{align}
0&=I(X_{1};X_{2}|T)-I(X_{1};X_{2}|Y,Z,T)\nonumber\\
&=I(X_{1};Y,Z|T)-I(X_{1};Y,Z|X_{2},T)\nonumber\\
&=I(X_{1},X_{2};Y,Z|T)-I(X_{1};Y,Z|X_{2},T)-I(X_{2};Y,Z|X_{1},T)\label{DBPC10}
\end{align}
Using (\ref{DBPC10}), it is clear that the sum of constraints
(\ref{DBPC1}) and (\ref{DBPC2}) is at least as strong as the
constraint (\ref{DBPC6}). This shows that (\ref{DBPC6}) is
redundant for the class of parallel channels where
$I(X_{1};X_{2}|Y,Z,T)=0$.

\section{Binary Additive Noisy MAC-FB}
In this paper, we will consider a binary-input additive noisy MAC
given by
\begin{align}
Y&=X_{1}+X_{2}+N\label{B1}
\end{align}
where $N$ is binary, uniform over $\{0,1\}$ and is independent of
$X_{1}$ and $X_{2}$. The channel output $Y$ takes values from the
set $\mathcal{Y}=\{0,1,2,3\}$. This channel does not fall into any
class of MAC for which the feedback capacity region is known. This
channel was also considered by Kramer in \cite{Kramer:thesis,
Kramer:2003} where it was shown that the Cover-Leung achievable
rate is strictly sub-optimal for the sum-rate.

We select a parallel channel $p^{+}(z|x_{1},x_{2},y)$ such that
$I(X_{1};X_{2}|Y,Z,T)=0$. By (\ref{DBPC9}), this will imply
$I(X_{1};X_{2}|T)=0$, and hence only distributions of the type
$p(t,x_{1},x_{2})=p(t)p(x_{1}|t)p(x_{2}|t)$ will be allowed. By
doing so, we restrict the set of allowable input distributions to
be the smallest permitted by $\mathcal{DB}_{PC}$, although we pay
a penalty due to the positive ``leak" terms $I(X_{1};Z|Y,X_{2},T)$
and $I(X_{2};Z|Y,X_{1},T)$.

Two simple choices of $Z$ which yield $I(X_{1};X_{2}|Y,Z,T)=0$ are
$Z=X_{1}$ and $Z=X_{2}$. For each of these choices, the
corresponding outer bounds are,
\begin{align}
\mathcal{DB}_{PC}^{(1)}=\Big\{(R_{1},R_{2}): \hspace{0.05in}&R_{1}\leq I(X_{1};Y|X_{2},T)+H(X_{1}|Y,X_{2},T)\label{B2}\\
&R_{2}\leq I(X_{2};Y|X_{1},T)\label{B3}\\
&R_{1}\leq I(X_{1};Y|X_{2})\label{B4}\\
&R_{1}+R_{2}\leq  I(X_{1},X_{2};Y)\label{B5}\Big\}
\end{align}
and
\begin{align}
\mathcal{DB}_{PC}^{(2)}=\Big\{(R_{1},R_{2}): \hspace{0.05in}&R_{1}\leq I(X_{1};Y|X_{2},T)\label{B6}\\
&R_{2}\leq I(X_{2};Y|X_{1},T)+H(X_{2}|Y,X_{1},T)\label{B7}\\
&R_{2}\leq I(X_{2};Y|X_{1})\label{B8}\\
&R_{1}+R_{2}\leq  I(X_{1},X_{2};Y)\label{B9}\Big\}
\end{align}
where both $\mathcal{DB}_{PC}^{(1)}$ and $\mathcal{DB}_{PC}^{(2)}$
are evaluated over the set of input distributions of the form
$p(t,x_{1},x_{2})=p(t)p(x_{1}|t)p(x_{2}|t)$.

For the binary additive noisy MAC-FB in consideration which is
given in (\ref{B1}), the following equalities hold for any
distribution of the form
$p(t,x_{1},x_{2})=p(t)p(x_{1}|t)p(x_{2}|t)$,
\begin{align}
H(X_{1}|Y,X_{2},T)&=\frac{1}{2}H(X_{1}|T)\label{B10}\\
H(X_{2}|Y,X_{1},T)&=\frac{1}{2}H(X_{2}|T)\label{B11}
\end{align}
Using (\ref{B10}) and (\ref{B11}), we can simplify
$\mathcal{DB}_{PC}^{(1)}$ and $\mathcal{DB}_{PC}^{(2)}$ as,
\begin{align}
\mathcal{DB}_{PC}^{(1)}=\Big\{(R_{1},R_{2}): \hspace{0.05in}&R_{1}\leq \min\left(I(X_{1};Y|X_{2}),H(X_{1}|T)\right)\label{B12}\\
&R_{2}\leq \frac{1}{2}H(X_{2}|T)\label{B13}\\
&R_{1}+R_{2}\leq  I(X_{1},X_{2};Y)\label{B14}\Big\}
\end{align}
and
\begin{align}
\mathcal{DB}_{PC}^{(2)}=\Big\{(R_{1},R_{2}): \hspace{0.05in}&R_{1}\leq \frac{1}{2}H(X_{1}|T)\label{B15}\\
&R_{2}\leq \min\left(I(X_{2};Y|X_{1}),H(X_{2}|T)\right)\label{B16}\\
&R_{1}+R_{2}\leq  I(X_{1},X_{2};Y)\label{B17}\Big\}
\end{align}
where both bounds are evaluated over the set of distributions of
the form $p(t,x_{1},x_{2})=p(t)p(x_{1}|t)p(x_{2}|t)$ and the
auxiliary random variable $T$ is subject to a cardinality
constraint of $|\mathcal{T}|\leq
|\mathcal{X}_{1}||\mathcal{X}_{2}|+3$. The evaluation of the above
outer bounds is rather cumbersome because for binary inputs, the
bound on $|\mathcal{T}|$ is $|\mathcal{T}|\leq 7$. To the best of
our knowledge, no one has been able to conduct an exhaustive
search over an auxiliary random variable whose cardinality is
larger than $4$. In Section $8$, we will obtain an alternate
characterization for our outer bounds using composite functions
and their properties. For that, we will first develop some useful
properties of composite functions in the next section.

A valid outer bound is given by the intersection of
$\mathcal{DB}_{PC}^{(1)}$ and $\mathcal{DB}_{PC}^{(2)}$,
\begin{align}
\mathcal{DB}_{PC}&=\mathcal{DB}_{PC}^{(1)}\bigcap\mathcal{DB}_{PC}^{(2)}\label{B18}
\end{align}
We will show that this outer bound is strictly smaller than the
cut-set bound at all points on the capacity region where feedback
increases capacity.

\section{Composite Functions and Their Properties}
Before obtaining a characterization of our outer bounds, we will
define a composite function and prove two lemmas regarding its
properties. These lemmas will be essential in obtaining simple
characterizations for our outer bounds and the Cover-Leung
achievable rate region. Throughout the paper, we will refer to the
entropy function as $h^{(k)}(s_{1},\ldots s_{k})$ which is defined
as,
\begin{align}
h^{(k)}(s_{1},\ldots,s_{k})&=-\sum_{i=1}^{k}s_{i}\text{log}(s_{i})\label{CF1}
\end{align}
for $s_{i}\geq 0$, $i=1\ldots,k$, and $\sum_{i=1}^{k}s_{i}=1$,
where all logarithms are to the base $2$. We will denote
$h^{(2)}(s)$ simply as $h(s)$. To characterize our bounds, we will
make use of the following function
\begin{align}
\phi(s)&=\left\{
           \begin{array}{ll}
             \frac{1-\sqrt{1-2s}}{2}, & \hbox{for $0\leq s\leq 1/2$} \\
             \frac{1-\sqrt{2s-1}}{2}, & \hbox{for $1/2 < s\leq 1$}
           \end{array}
         \right.\label{CF2}
\end{align}
It was shown in \cite{WillemsonFB:1984} that the composite function
$h(\phi(s))$ is symmetric around $s=1/2$ and concave in $s$ for
$0\leq s\leq 1$. The functions $\phi(s)$ and $h(\phi(s))$ are
illustrated in Figure $2$. From the definition of $\phi(s)$ in
(\ref{CF2}) it is clear that for any $s\in [0,1]$, the function
$\phi(s)$ satisfies the following property
\begin{align}
\phi(2s(1-s))&=\min(s,1-s)\label{CF3}
\end{align}
As a consequence, the following holds as well
\begin{align}
h(\phi(2s(1-s)))&=h(s)\label{CF4}
\end{align}
For any $s \in [0,1]$, the following holds from the definition of
$\phi(s)$,
\begin{align}
s&=\left\{%
\begin{array}{ll}
    \phi(2s(1-s)), & \hbox{$0\leq s \leq \frac{1}{2}$} \\
    1-\phi(2s(1-s)), & \hbox{$\frac{1}{2} < s \leq 1$}
    \label{CF5}
\end{array}%
\right.
\end{align}
For any $x \in[0,\frac{1}{2}]$ and $y \in[0,\frac{1}{2}]$, let us
define a function
\begin{align}
f(x,y)&\triangleq \phi(x)+\phi(y)-2\phi(x)\phi(y)\label{52}\\
&=\frac{1-\sqrt{(1-2x)(1-2y)}}{2}\label{53}
\end{align}
From the above definition, it is clear that the function $f(x,y)$
lies in the range $[0,\frac{1}{2}]$.
\begin{figure}[t]
  \centering
  \centerline{\epsfig{figure=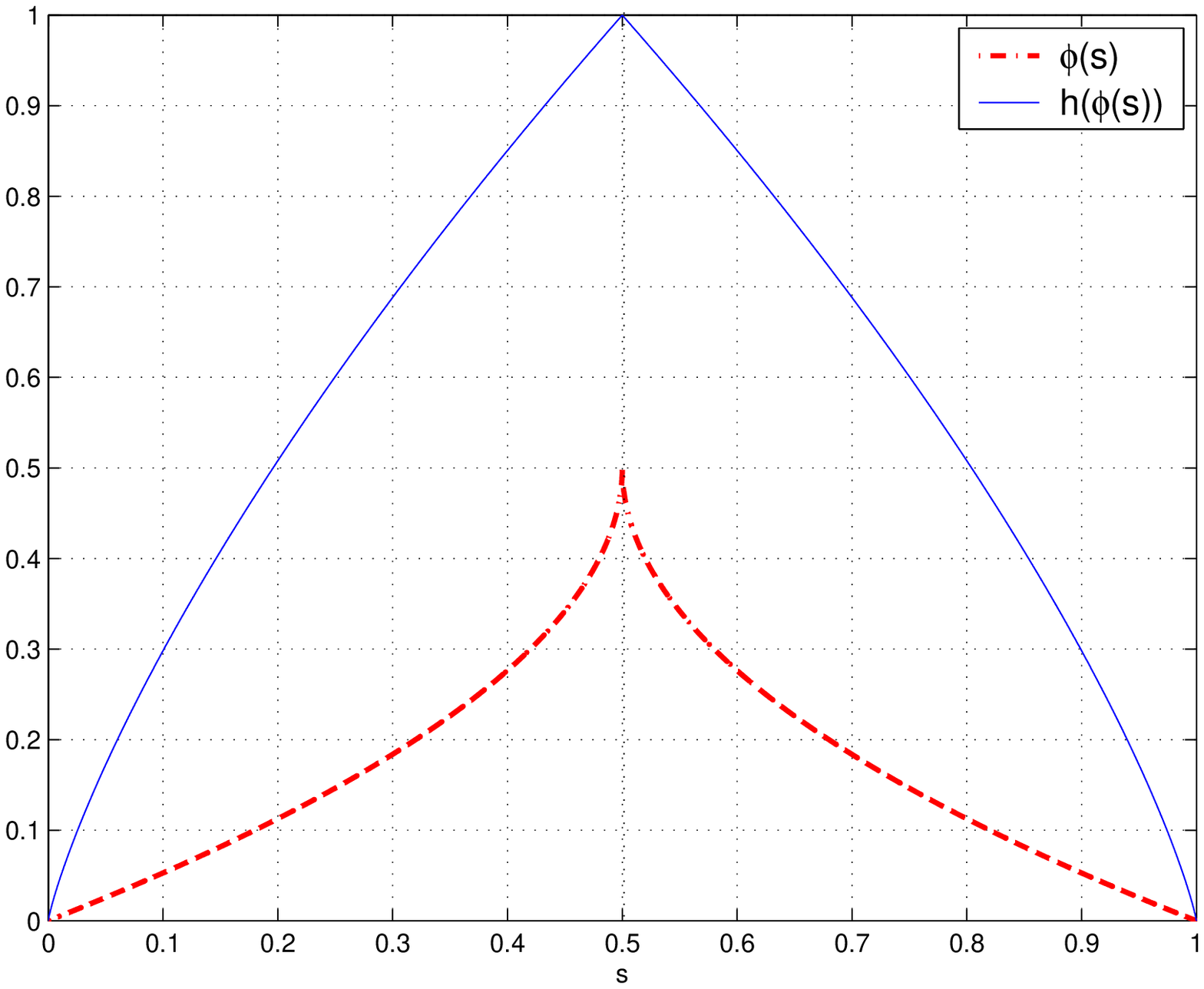,width=8cm}}
  { Figure 2: Functions $\phi(s)$ and $h(\phi(s))$.}\medskip
\end{figure}
\begin{Lem} The variable
\begin{align}
v&=s_{1}+s_{2}-2s_{1}s_{2}\label{54}
\end{align}
is always lower bounded by $f(2s_{1}(1-s_{1}),2s_{2}(1-s_{2}))$ for
any $s_{1}\in[0,1],s_{2}\in[0,1]$.
\end{Lem}
\begin{proof} We will prove this lemma by considering all four possible
cases.
\begin{enumerate}
    \item If
    $s_{1}\in[0,\frac{1}{2}],s_{2}\in[0,\frac{1}{2}]$, then from (\ref{CF5}), $s_{1}=\phi(2s_{1}(1-s_{1}))$,
    $s_{2}=\phi(2s_{2}(1-s_{2}))$ and hence
    \begin{align}
    v=f(2s_{1}(1-s_{1}),2s_{2}(1-s_{2}))
    \end{align}
    \item If
    $s_{1}\in[\frac{1}{2},1],s_{2}\in[\frac{1}{2},1]$, then from (\ref{CF5}), $s_{1}=1-\phi(2s_{1}(1-s_{1}))$,
    $s_{2}=1-\phi(2s_{2}(1-s_{2}))$ and hence
    \begin{align}
    v=f(2s_{1}(1-s_{1}),2s_{2}(1-s_{2}))
    \end{align}
    \item If
    $s_{1}\in[0,\frac{1}{2}],s_{2}\in[\frac{1}{2},1]$, then from (\ref{CF5}), $s_{1}=\phi(2s_{1}(1-s_{1}))$,
    $s_{2}=1-\phi(2s_{2}(1-s_{2}))$ and hence
    \begin{align}
    v&=1-f(2s_{1}(1-s_{1}),2s_{2}(1-s_{2}))\nonumber\\
         &\overset{(a)}\geq f(2s_{1}(1-s_{1}),2s_{2}(1-s_{2}))
    \end{align} where $(a)$ follows by the fact that $f(2s_{1}(1-s_{1}),2s_{2}(1-s_{2}))\leq
    \frac{1}{2}$.
    \item If
    $s_{1}\in[\frac{1}{2},1],s_{2}\in[0,\frac{1}{2}]$, then from (\ref{CF5}), $s_{1}=1-\phi(2s_{1}(1-s_{1}))$,
    $s_{2}=\phi(2s_{2}(1-s_{2}))$ and hence
    \begin{align}
    v&=1-f(2s_{1}(1-s_{1}),2s_{2}(1-s_{2}))\nonumber\\
         &\overset{(b)}\geq f(2s_{1}(1-s_{1}),2s_{2}(1-s_{2}))
    \end{align} where $(b)$ follows by the fact that $f(2s_{1}(1-s_{1}),2s_{2}(1-s_{2}))\leq
    \frac{1}{2}$.
\end{enumerate}
Thus, for any pair $(s_{1},s_{2})$, where $s_{1}\in[0,1]$,
$s_{2}\in [0,1]$, we have shown that $v\geq
f(2s_{1}(1-s_{1}),2s_{2}(1-s_{2}))$.
\end{proof}

\begin{Lem} The function
$f(x,y)$ is jointly convex in $(x,y)$ for $0\leq x\leq
\frac{1}{2},0\leq y\leq \frac{1}{2} $.\end{Lem}

\begin{proof}
Showing that the function $f(x,y)$ is jointly convex in $(x,y)$ is
equivalent to showing that the Hessian matrix, $H$ of $f(x,y)$ is
positive semi-definite, which is equivalent to showing that the
eigenvalues of $H$ are non-negative. The Hessian matrix, $H$, of
$f(x,y)$ is
\begin{align}
H&=\left(
        \begin{array}{cc}
          \frac{\sqrt{1-2y}}{2(1-2x)^{3/2}} & \frac{-1}{2\sqrt{(1-2x)(1-2y)}} \\
          \frac{-1}{2\sqrt{(1-2x)(1-2y)}} & \frac{\sqrt{1-2x}}{2(1-2y)^{3/2}} \\
        \end{array}
      \right)
\end{align}
The two eigenvalues of $H$ are
\begin{align}
\lambda_{1}&=0\nonumber\\
\lambda_{2}&=\frac{1}{2}\Bigg(\frac{\sqrt{1-2y}}{(1-2x)^{3/2}}+\frac{\sqrt{1-2x}}{(1-2y)^{3/2}}\Bigg)
\end{align}
which are non-negative for all $0\leq x \leq \frac{1}{2}$ and $0\leq
y \leq \frac{1}{2}$, thus completing the proof.\end{proof}

\section{Evaluation of the Dependence Balance Outer Bound}
We will now return to the characterization of our upper bounds
$\mathcal{DB}_{PC}^{(1)}$ and $\mathcal{DB}_{PC}^{(2)}$. Let the
cardinality of the auxiliary random variable $T$ be fixed and
arbitrary, say $|\mathcal{T}|$. Then, the joint distribution
$p(t)p(x_{1}|t)p(x_{2}|t)$ can be described by the following
variables:
\begin{align}
q_{1t}&=\text{Pr}(X_{1} =0|T=t), \hspace{0.2in} t=1,\ldots,|\mathcal{T}|\nonumber\\
q_{2t}&=\text{Pr}(X_{2} =0|T=t), \hspace{0.2in} t=1,\ldots,|\mathcal{T}|\nonumber\\
p_{t}&=\text{Pr}(T=t), \hspace{0.75in}
t=1,\ldots,|\mathcal{T}|\label{EV1}
\end{align}
We will characterize our outer bounds in terms of three variables
$u_{1}$, $u_{2}$ and $u$ which are functions of
$p(t,x_{1},x_{2})$, and are defined as,
\begin{align}
u_{1}&=\sum_{t}p_{t}q_{1t}(1-q_{1t})=\sum_{t}p_{t}u_{1t}\label{EV2}\\
u_{2}&=\sum_{t}p_{t}q_{2t}(1-q_{2t})=\sum_{t}p_{t}u_{2t}\label{EV3}\\
u&=\sum_{t}p_{t}(q_{1t}+q_{2t}-2q_{1t}q_{2t})=\sum_{t}p_{t}u_{t}\label{EV4}
\end{align}
where we have defined
\begin{align}
u_{1t}&=q_{1t}(1-q_{1t})\label{EV5}\\
u_{2t}&=q_{2t}(1-q_{2t})\label{EV6}\\
u_{t}&=q_{1t}+q_{2t}-2q_{1t}q_{2t}\label{EV7}
\end{align}
It should be noted that since $0\leq q_{jt}\leq 1$, for $j=1,2$,
$t=1,\ldots,|\mathcal{T}|$, the variables $u_{1},u_{2},u_{1t}$ and
$u_{2t}$ all lie in the range $[0,\frac{1}{4}]$. Our outer bounds
$\mathcal{DB}_{PC}^{(1)}$ and $\mathcal{DB}_{PC}^{(2)}$ are
comprised of the following information theoretic entities:
\begin{enumerate}
\item $H(X_{1}|T)$, $H(X_{2}|T)$
\item $I(X_{1};Y|X_{2})$, $I(X_{2};Y|X_{1})$
\item $I(X_{1},X_{2};Y)$.
\end{enumerate}
We will first obtain upper bounds for each one of these entities
individually in terms of $(u_{1},u_{2},u)$.

We upper bound $H(X_{1}|T)$ as follows,
\begin{align}
H(X_{1}|T) &=\sum_{t}p_{t}h(q_{1t})\label{EV8}\\
&=\sum_{t}p_{t}h(\phi(2q_{1t}(1-q_{1t})))\label{EV9}\\
&=\sum_{t}p_{t}h(\phi(2u_{1t}))\label{EV10}\\
&\leq h(\phi(2u_{1}))\label{EV11}
\end{align}
where (\ref{EV9}) follows due to (\ref{CF4}), (\ref{EV10}) follows
from (\ref{EV5}), and (\ref{EV11}) follows from the fact that
$h(\phi(s))$ is concave in $s$ and the application of Jensen's
inequality \cite{Cover:book}. Using a similar set of inequalities
for $H(X_{2}|T)$, we obtain
\begin{align}
H(X_{2}|T)&\leq h(\phi(2u_{2}))\label{EV12}
\end{align}

We will now upper bound $I(X_{1};Y|X_{2})$ in terms of the
variable $u$. For this purpose, let us first define
\begin{align}
a&=P_{X_{1}X_{2}}(0,0)=\sum_{t}p_{t}q_{1t}q_{2t}\label{EV13}\\
b&=P_{X_{1}X_{2}}(0,1)=\sum_{t}p_{t}q_{1t}(1-q_{2t})\label{EV14}\\
c&=P_{X_{1}X_{2}}(1,0)=\sum_{t}p_{t}(1-q_{1t})q_{2t}\label{EV15}\\
d&=P_{X_{1}X_{2}}(1,1)=1-a-b-c.\label{EV16}
\end{align}
We now proceed as,
\begin{align}
I(X_{1};Y|X_{2})
=& \hspace{0.05in}H(Y|X_{2})-H(Y|X_{1},X_{2})\label{EV17}\\
=& \hspace{0.05in} H(Y|X_{2})-1\label{EV18}\\
=&\hspace{0.05in}(a+c)h^{(3)}\left(\frac{a}{2(a+c)},\frac{1}{2},\frac{c}{2(a+c)}\right)\nonumber\\
&+(b+d)h^{(3)}\left(\frac{b}{2(b+d)},\frac{1}{2},\frac{d}{2(b+d)}\right)-1\label{EV19}\\
\leq&\hspace{0.05in} h^{(3)}\left(\frac{a+d}{2},\frac{1}{2},\frac{b+c}{2}\right)-1\label{EV20}\\
=&\hspace{0.05in}\frac{1}{2}h(b+c)\label{EV21}\\
=&\hspace{0.05in}\frac{1}{2}h(u)\label{EV22}
\end{align}
where (\ref{EV20}) follows by the concavity of the entropy
function and the application of Jensen's inequality
\cite{Cover:book}. Using a similar set of inequalities, we also
have
\begin{align}
I(X_{2};Y|X_{1})&\leq\frac{1}{2}h(u)\label{EV23}
\end{align}

We will now obtain an upper bound on $I(X_{1},X_{2};Y)$. First note
that
\begin{align}
I(X_{1},X_{2};Y)&=H(Y)-H(Y|X_{1},X_{2})\label{EV24}\\
&=h^{(4)}(P_{Y}(0),P_{Y}(1),P_{Y}(2),P_{Y}(3))-1\label{EV25}
\end{align}
where
\begin{align}
P_{Y}(0)&=\sum_{t}p_{t}q_{1t}q_{2t}/2\label{EV26}\\
P_{Y}(1)&=\sum_{t}p_{t}\big(q_{1t}+q_{2t}-q_{1t}q_{2t}\big)/2\label{EV27}\\
P_{Y}(2)&=\sum_{t}p_{t}\big(1-q_{1t}q_{2t}\big)/2\label{EV28}\\
P_{Y}(3)&=\sum_{t}p_{t}(1-q_{1t})(1-q_{2t})/2\label{EV29}
\end{align}
Using the following fact,
\begin{align}
h^{(4)}(\alpha,\beta,\gamma,\theta)&=\frac{1}{2}h^{(4)}(\alpha,\beta,\gamma,\theta)+\frac{1}{2}h^{(4)}(\theta,\gamma,\beta,\alpha)\label{EV30}\\
&\leq h^{(4)}\left(\frac{\alpha+\theta}{2},\frac{\beta+\gamma}{2},\frac{\beta+\gamma}{2},\frac{\alpha+\theta}{2}\right)\label{EV31}\\
&= h\left(\alpha+\theta\right)+h\left(\frac{1}{2}\right)\label{EV32}\\
&= h\left(1-(\beta+\gamma)\right)+1\label{EV33}
\end{align}
where (\ref{EV31}) follows by the concavity of the entropy
function and the application of Jensen's inequality
\cite{Cover:book}, we now obtain an upper bound on
$I(X_{1},X_{2};Y)$ by continuing from (\ref{EV25}),
\begin{align}
I(X_{1},X_{2};Y)&=h^{(4)}(P_{Y}(0),P_{Y}(1),P_{Y}(2),P_{Y}(3))-1\label{EV34}\\
&\leq h\left(1-(P_{Y}(1)+P_{Y}(2))\right)+h\left(\frac{1}{2}\right)-1\label{EV35}\\
&= h\left(\frac{1-u}{2}\right)\label{EV36}
\end{align}
where (\ref{EV35}) follows by (\ref{EV33}) and (\ref{EV36})
follows from the fact that $P_{Y}(1)+P_{Y}(2)=(1+u)/2$ using
(\ref{EV27}) and (\ref{EV28}),  where $u$ is as defined in
(\ref{EV4}).

\subsection{A Set of Feasible $(u_{1},u_{2},u)$: $\mathcal{P}$}

We have obtained upper bounds on the information theoretic
entities which comprise our outer bounds in terms of three
variables $u_{1},u_{2}$ and $u$. We will now give a feasible
region for these triples based on the structures of these
variables. First, note that for any $q_{1t}\in[0,1]$, the
following holds: $u_{1t}=q_{1t}(1-q_{1t})\leq \frac{1}{4}$.
Similarly, $u_{2t}=q_{2t}(1-q_{2t})\leq \frac{1}{4}$. Hence, we
have
\begin{align}
0&\leq u_{1}\leq \frac{1}{4}\label{P1}\\
0&\leq u_{2}\leq \frac{1}{4}\label{P2}
\end{align}
We now obtain a lower bound on $u$ as
\begin{align}
u&=\sum_{t}p_{t}u_{t}\label{P3}\\
&\geq \sum_{t}p_{t}f(2u_{1t},2u_{2t})\label{P4}\\
&\geq f\left(2\sum_{t}p_{t}u_{1t},2\sum_{t}p_{t}u_{2t}\right)\label{P5}\\
&=f(2u_{1},2u_{2})\label{P6}
\end{align}
where (\ref{P4}) follows by Lemma $1$ and (\ref{P5}) follows by
Lemma $2$ and the application of Jensen's inequality
\cite{Cover:book}. We now obtain another lower bound on $u$,
\begin{align}
u&=\sum_{t}p_{t}u_{t}\label{P7}\\
&=\sum_{t}p_{t}(q_{1t}+q_{2t}-2q_{1t}q_{2t})\label{P8}\\
&=\sum_{t}p_{t}(q_{1t}-q_{1t}^{2}+q_{2t}-q_{2t}^{2}+(q_{1t}-q_{2t})^{2})\label{P9}\\
&\geq \sum_{t}p_{t}(q_{1t}-q_{1t}^{2}+q_{2t}-q_{2t}^{2})\label{P10}\\
&=
\sum_{t}p_{t}q_{1t}(1-q_{1t})+\sum_{t}p_{t}q_{2t}(1-q_{2t})\label{P11}\\
&=u_{1}+u_{2}\label{P12}
\end{align}
Finally, we obtain an upper bound on $u$ in terms of $u_{1}$ and
$u_{2}$,
\begin{align}
u&=\sum_{t}p_{t}u_{t}\label{P13}\\
&=\sum_{t}p_{t}(q_{1t}+q_{2t}-2q_{1t}q_{2t})\label{P14}\\
&=\sum_{t}p_{t}(q_{1t}+q_{2t}-2q_{1t}q_{2t}+q_{1t}^{2}+(1-q_{2t})^{2}-q_{1t}^{2}-(1-q_{2t})^{2})\label{P15}\\
&\leq \sum_{t}p_{t}(q_{1t}+q_{2t}-2q_{1t}q_{2t}+q_{1t}^{2}+(1-q_{2t})^{2}-2q_{1t}(1-q_{2t}))\label{P16}\\
&=1-(u_{1}+u_{2})\label{P17}
\end{align}
where (\ref{P16}) follows by the inequality
$q_{1t}^{2}+(1-q_{2t})^{2}\geq 2q_{1t}(1-q_{2t})$.

By noting
\begin{align}
f(2u_{1},2u_{2})-(u_{1}+u_{2})&=\frac{1-\sqrt{(1-4u_{1})(1-4u_{2})}}{2}-(u_{1}+u_{2})\label{P18}\\
&=\frac{(1-4u_{1})+(1-4u_{2})-2\sqrt{(1-4u_{1})(1-4u_{2})}}{4}\label{P19}\\
&=\frac{(\sqrt{1-4u_{1}}-\sqrt{1-4u_{2}})^{2}}{4}\label{P20}\\
&\geq 0\label{P21}
\end{align}
and using (\ref{P6}), we note that the lower bound in (\ref{P12})
is redundant. Therefore, from (\ref{P6}) and (\ref{P17}), we have
the following feasible range for the variable $u$ in terms of
$u_{1}$ and $u_{2}$,
\begin{align}
f(2u_{1},2u_{2})\leq u \leq 1-(u_{1}+u_{2})\label{P22}
\end{align}
Combining (\ref{P1}), (\ref{P2}) and (\ref{P22}), a set of feasible
$(u_{1},u_{2},u)$ is given as follows,
\begin{align}
\mathcal{P}=\Big\{&(u_{1},u_{2},u): 0\leq u_{1}\leq
\frac{1}{4};0\leq u_{2}\leq \frac{1}{4}; f(2u_{1},2u_{2})\leq u \leq
1-(u_{1}+u_{2})\Big\}\label{P23}
\end{align}
It should be noted that the set $\mathcal{P}$ in (\ref{P23}) may
not necessarily be the smallest feasible set of all triples
$(u_{1},u_{2},u)$. Since we are interested in a maximization over
these set of triples, a possibly larger set $\mathcal{P}$
suffices.

\subsection{A Simple Characterization of $\mathcal{DB}_{PC}^{(1)}$ and $\mathcal{DB}_{PC}^{(2)}$}

Using the upper bounds on $H(X_{1}|T)$, $H(X_{2}|T)$,
$I(X_{1};Y|X_{2})$, $I(X_{2};Y|X_{1})$ and $I(X_{1},X_{2};Y)$ in
(\ref{EV11}), (\ref{EV12}), (\ref{EV22}), (\ref{EV23}) and
(\ref{EV36})  in terms of $(u_{1},u_{2},u)$ along with a feasible
set of triples $\mathcal{P}$ in (\ref{P23}), we obtain the
following two outer bounds on the capacity region of the binary
additive noisy MAC-FB, starting from (\ref{B12})-(\ref{B14}) and
(\ref{B15})-(\ref{B17}),
\begin{align}
\mathcal{DB}_{PC}^{(1)}=\bigcup_{(u_{1},u_{2},u)\in
\mathcal{P}}\Bigg\{(R_{1},R_{2}): \hspace{0.05in}&R_{1}\leq
\min\left(\frac{1}{2}h(u),h(\phi(2u_{1}))\right)\nonumber\\
&R_{2}\leq \frac{1}{2}h(\phi(2u_{2}))\nonumber\\
&R_{1}+R_{2}\leq h\left(\frac{1-u}{2}\right)\Bigg\}\label{DB1}
\end{align}
and
\begin{align}
\mathcal{DB}_{PC}^{(2)}=\bigcup_{(u_{1},u_{2},u)\in
\mathcal{P}}\Bigg\{(R_{1},R_{2}): \hspace{0.05in}&R_{1}\leq \frac{1}{2}h(\phi(2u_{1}))\nonumber\\
&R_{2}\leq
\min\left(\frac{1}{2}h(u),h(\phi(2u_{2}))\right)\nonumber\\
&R_{1}+R_{2}\leq h\left(\frac{1-u}{2}\right)\Bigg\}\label{DB2}
\end{align}
We will plot these outer bounds and their intersection in Figure
$4$. In next section, we will explicitly characterize our upper
bounds for the symmetric-rate point on the capacity region of the
binary additive noisy MAC-FB in consideration.

\section{Explicit Characterization of the Symmetric-rate Upper Bound}
For the binary additive noisy MAC-FB in consideration, it was
shown by Kramer \cite{Kramer:thesis} that the symmetric-rate
cut-set bound is $0.45915$ bits/transmission. It was also shown in
\cite{Kramer:thesis} that the Cover-Leung achievable
symmetric-rate is $0.43621$ bits/transmission and it was improved
to $0.43879$ bits/transmission by using superposition coding and
binning with code trees. For completeness and comparison with
existing bounds, we will first completely characterize our outer
bound for the symmetric-rate by providing the input distribution
$p(t)p(x_{1}|t)p(x_{2}|t)$ which achieves it. By symmetric-rate we
mean a rate $R$ such that the rate pair $(R,R)$ lies in the
capacity region of MAC-FB. For the symmetric-rate, both
$\mathcal{DB}_{PC}^{(1)}$ and $\mathcal{DB}_{PC}^{(2)}$ will yield
the same upper bound. Hence, we will focus on
$\mathcal{DB}_{PC}^{(1)}$. Using (\ref{DB1}), we are interested in
obtaining the largest $R$ over all $(u_{1},u_{2},u)\in
\mathcal{P}$ such that
\begin{align}
R&\leq \min\left(\frac{1}{2}h(u),h(\phi(2u_{1}))\right)\label{S1}\\
R&\leq \frac{1}{2}h(\phi(2u_{2}))\label{S2}\\
2R&\leq h\left(\frac{1-u}{2}\right)\label{S3}
\end{align}
We will show that a seemingly weaker version of the above bound
will improve upon the symmetric-rate cut-set  bound. We will also
show that the weaker bound is in fact the same as the above bound,
and its sole purpose is the simplicity of evaluation and insight
into the input distribution that attains it. We first obtain a
weakened version of (\ref{S1}) as
\begin{align}
R \leq \min\left(\frac{1}{2}h(u),h(\phi(2u_{1}))\right) \leq
h(\phi(2u_{1}))\label{S4}
\end{align}
Next, consider (\ref{S3})
\begin{align}
2R&\leq h\left(\frac{1-u}{2}\right)\label{S5}\\
&=h\left(\frac{1}{2}-\frac{u}{2}\right)\label{S6}\\
&\leq h\left(\frac{1}{2}-\frac{f(2u_{1},2u_{2})}{2}\right)\label{S7}
\end{align}
where (\ref{S7}) follows from (\ref{P6}) and the fact that the
binary entropy function $h(s)$ is monotonically increasing in $s$
for $s\in [0,\frac{1}{2}]$. Combining (\ref{S2}), (\ref{S4}) and
(\ref{S7}), we are interested in the largest $R$ such that
\begin{align}
R &\leq \max_{u_{1},u_{2} \in [0,\frac{1}{4}]}
\min\left(h(\phi(2u_{1})), \frac{1}{2}h(\phi(2u_{2})),
\frac{1}{2}h\left(\frac{1}{2}-\frac{f(2u_{1},2u_{2})}{2}\right)\right)\label{S8}
\end{align}
We note that this upper bound on the symmetric-rate depends only
on $u_{1}$ and $u_{2}$, and therefore, we replace the feasible set
$\mathcal{P}$ with $u_{1},u_{2}\in [0,\frac{1}{4}]$.

We know that $h(\phi(s))$ is concave in $s$ for $s \in[0,1]$.
Hence, it follows that both $h(\phi(2u_{1}))$ and
$\frac{1}{2}h(\phi(2u_{2}))$ are concave in $u_{1}$ and $u_{2}$,
respectively, and hence concave in the pair $(u_{1},u_{2})$. We
also have the following lemma.
\begin{Lem}
The function
\begin{align}
g(u_{1},u_{2})&=\frac{1}{2}h\left(\frac{1-f(2u_{1},2u_{2})}{2}\right)\label{Lem3}
\end{align}
is monotonically decreasing and jointly concave in the pair
$(u_{1},u_{2})$ for $u_{1}, u_{2}\in[0,\frac{1}{4}]$.\end{Lem}
\begin{proof} It suffices to show that for a fixed $u_{2}$, the function $g(u_{1},u_{2})$
is monotonically decreasing in $u_{1}$. Substituting the value of
$f(2u_{1},2u_{2})$, we have
\begin{align}
g(u_{1},u_{2})&=\frac{1}{2}h\left(\frac{1-(\phi(2u_{1})+\phi(2u_{2})-2\phi(2u_{1})\phi(2u_{2}))}{2}\right)\label{L1}\\
&=\frac{1}{2}h\left(\frac{1}{2}-\frac{\phi(2u_{2})}{2}-\frac{\phi(2u_{1})(1-2\phi(2u_{2}))}{2}\right)\label{L2}
\end{align}
Now using the fact that $\phi(2s)$ is increasing in $s$ for $s\in[0,
\frac{1}{4}]$, we have that for $u_{1}^{'}\geq u_{1}$,
$\phi(2u_{1}^{'})\geq \phi(2u_{1})$. Moreover, the following holds
\begin{align}
\frac{\phi(2u_{1}^{'})(1-2\phi(2u_{2}))}{2}\geq
\frac{\phi(2u_{1})(1-2\phi(2u_{2}))}{2}\label{L3}
\end{align}
since $\phi(2u_{2})\leq \frac{1}{2}$. Now using the above inequality
along with the fact that the binary entropy function $h(s)$ is
increasing for $0 \leq s \leq \frac{1}{2}$, we have that for
$u_{1}^{'}\geq u_{1}$,
\begin{align}
&\frac{1}{2}h\left(\frac{1}{2}-\frac{\phi(2u_{2})}{2}-\frac{\phi(2u_{1})(1-2\phi(2u_{2}))}{2}\right)\geq
\frac{1}{2}h\left(\frac{1}{2}-\frac{\phi(2u_{2})}{2}-\frac{\phi(2u_{1}^{'})(1-2\phi(2u_{2}))}{2}\right)\label{L4}
\end{align}
This shows that for a fixed $u_{2}$, the function $g(u_{1},u_{2})$
is monotonically decreasing in $u_{1}$. As the function is
symmetric in $u_{1}$ and $u_{2}$, the monotonicity of
$g(u_{1},u_{2})$ in $(u_{1},u_{2})$ follows.

To show the concavity of $g(u_{1},u_{2})$ in the pair
$(u_{1},u_{2})$, we first note from Lemma $2$ that
$f(2u_{1},2u_{2})$ is jointly convex in the pair $(u_{1},u_{2})$. We
define another function
\begin{align}
\xi(u_{1},u_{2})&=\frac{1-f(2u_{1},2u_{2})}{2}\label{L5}
\end{align}
Note that $\xi(u_{1},u_{2})$ is jointly concave in the pair
$(u_{1},u_{2})$. Furthermore, the binary entropy function $h(s)$
is concave and nondecreasing for $s\in [0,\frac{1}{2}]$. Hence,
rewriting the function $g(u_{1},u_{2})$ as a composition of two
functions, we obtain
\begin{align}
g(u_{1},u_{2})&=\frac{1}{2}h(\xi(u_{1},u_{2}))\label{L6}
\end{align}
From the theory of composite functions \cite{Boyd:book}, we know
that a composite function $f_{1}(f_{2}(s))$ is concave in $s$ if
$f_{1}(.)$ is concave and nondecreasing and $f_{2}(s)$ is concave in
$s$. Identifying $f_{1}(.)$ with $h(.)$ and $f_{2}(u_{1},u_{2})$
with $\xi(u_{1},u_{2})$, the concavity of $g(u_{1},u_{2})$ in the
pair $(u_{1},u_{2})$ is established.
\end{proof}

Therefore, all three functions in the $\min(.)$ in (\ref{S8}) are
concave in $(u_{1},u_{2})$. Invoking the fact that the minimum of
concave functions is concave, we conclude that the maximum in
(\ref{S8}) is unique. We will now show that the unique pair
$(u_{1}^{*},u_{2}^{*})$ that attains this maximum satisfies the
property that
$h(\phi(2u_{1}^{*}))=\frac{1}{2}h(\phi(2u_{2}^{*}))=g(u_{1}^{*},u_{2}^{*})$.

For this purpose, we first characterize those pairs
$(\tilde{u}_{1},\tilde{u}_{2})$ such that the following holds,
\begin{align}
h(\phi(2\tilde{u}_{1}))=\frac{1}{2}h(\phi(2\tilde{u}_{2}))=g(\tilde{u}_{1},\tilde{u}_{2})\label{S9}
\end{align}
By using (\ref{S9}), we obtain two equations for $\tilde{u}_{1}$
and $\tilde{u}_{2}$, as
\begin{align}
h(\phi(2\tilde{u}_{1}))&=\frac{1}{2}h\left(\frac{1-\phi(2\tilde{u}_{1})}{3-2\phi(2\tilde{u}_{1})}\right)\label{S10}\\
\phi(2\tilde{u}_{2})&=\frac{1-\phi(2\tilde{u}_{1})}{3-2\phi(2\tilde{u}_{1})}\label{S11}
\end{align}
From (\ref{S10}), one can see that $2\tilde{u}_{1}$ is the unique
solution $s \in [0,\frac{1}{2}]$ of the equation
\begin{align}
h(\phi(s))&=\frac{1}{2}h\left(\frac{1-\phi(s)}{3-2\phi(s)}\right)\label{S12}
\end{align}
Obtaining the optimal $\tilde{u}_{1}$ from the above equation is
illustrated in Figure $3$. The unique solutions
$(\tilde{u}_{1},\tilde{u}_{2})$ of (\ref{S10}) and (\ref{S11}) are
\begin{align}
\tilde{u}_{1}&=0.086063, \qquad\tilde{u}_{2}=0.218333\label{S14}
\end{align}
We will now show that this pair $(\tilde{u}_{1},\tilde{u}_{2})$
yields the maximum in (\ref{S8}).

Returning to the maximization problem (\ref{S8}), first denote
$\mathcal{S}$ as the region of allowable $(u_{1},u_{2})$,
\begin{align}
\mathcal{S}&=\Big\{(u_{1},u_{2}): 0\leq u_{1}\leq \frac{1}{4}; 0\leq
u_{2}\leq \frac{1}{4}\Big\}\label{S15}
\end{align}
Also define a subset of this region
\begin{align}
\mathcal{\tilde{S}}&=\Big\{(u_{1},u_{2}): u_{1} \in
(\tilde{u}_{1},\frac{1}{4}]; u_{2} \in
(\tilde{u}_{2},\frac{1}{4}]\Big\}\label{S16}
\end{align}
where $(\tilde{u}_{1},\tilde{u}_{2})$ is given by (\ref{S14}). We
will now show that the pair $(\tilde{u}_{1},\tilde{u}_{2})$ yields
the solution of the maximization problem in (\ref{S8}). Consider
the following two cases,
\begin{enumerate}
    \item  If $(u_{1},u_{2}) \in \mathcal{\tilde{S}}$, then
    by Lemma $3$, we have that $g(u_{1},u_{2})\leq
    g(\tilde{u}_{1},\tilde{u}_{2})$, using which we obtain,
    \begin{align}
    \min\left(h(\phi(2u_{1})),\frac{1}{2}h(\phi(2u_{2})),g(u_{1},u_{2})\right)&\leq g(u_{1},u_{2})\leq
    g(\tilde{u}_{1},\tilde{u}_{2})\label{S17}
    \end{align}
    \item  If $(u_{1},u_{2}) \in \mathcal{S}\setminus
    \mathcal{\tilde{S}}$, we either have $u_{1}\leq \tilde{u}_{1}$ or
    $u_{2} \leq \tilde{u}_{2}$ or both. Using this along with the fact that $h(\phi(2s))$ is monotonically increasing in $s$ for $s\in [0,\frac{1}{4}]$, we obtain
    \begin{align}
    \min\left(h(\phi(2u_{1})),\frac{1}{2}h(\phi(2u_{2})),g(u_{1},u_{2})\right)&\leq
    h(\phi(2\tilde{u}_{1}))\label{S18}
    \end{align}
\end{enumerate}
\begin{figure}[t]
  \centering
  \centerline{\epsfig{figure=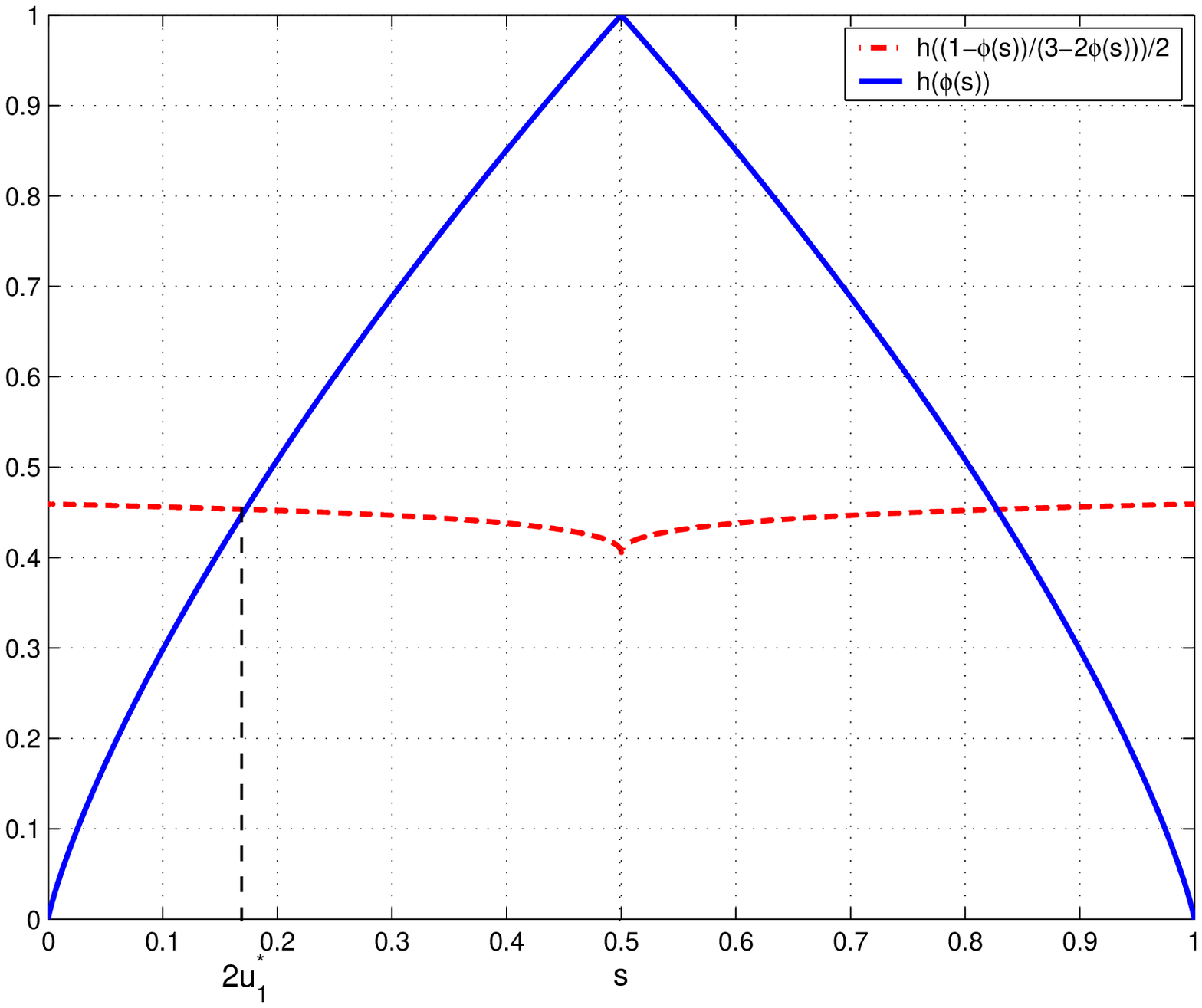,width=8cm}}
  { Figure $3$: Characterization of the optimal $u_{1}^{*}$.}\medskip
\end{figure}
The above two cases show the following,
\begin{align}
\max_{u_{1}\in [0,\frac{1}{4}], u_{2}\in [0,\frac{1}{4}]} \min
\left(h(\phi(2u_{1})),\frac{1}{2}h(\phi(2u_{2})),g(u_{1},u_{2})\right)
&=
h(\phi(2\tilde{u}_{1}))\\
&= \frac{1}{2}h(\phi(2\tilde{u}_{2}))\\
&= g(\tilde{u}_{1},\tilde{u}_{2})
\end{align}
Thus, the maximum in (\ref{S8}) is obtained at
$(u_{1}^{*},u_{2}^{*})=(\tilde{u}_{1},\tilde{u}_{2})$. We now
obtain a distribution $p(t)p(x_{1}|t)p(x_{2}|t)$ which attains
this symmetric-rate upper bound. Fix $T$ to be binary, and select
the involved probabilities as
\begin{align}
p_{0}=p_{1}&=\frac{1}{2}\\
q_{10}=1-q_{11}&=\phi(2u_{1}^{*})\\
q_{20}=1-q_{21}&=\phi(2u_{2}^{*})
\end{align}
The reason for constructing such an input distribution is that, at
this specific distribution, we have the following exact
equalities,
\begin{align}
H(X_{1}|T)&=h(\phi(2u_{1}^{*}))\\
\frac{1}{2}H(X_{2}|T)&=\frac{1}{2}h(\phi(2u_{2}^{*}))\\
\frac{1}{2}I(X_{1},X_{2};Y)&=g(u_{1}^{*},u_{2}^{*})
\end{align}
and we achieve the outer bound we developed with equality.
Substituting the values of $(u_{1}^{*},u_{2}^{*})$, we obtain a
distribution given by,
\begin{align}
p_{0}=p_{1}&=\frac{1}{2}\label{S20}\\
q_{10}=1-q_{11}&=0.095109\label{S21}\\
q_{20}=1-q_{21}&=0.322050\label{S22}
\end{align}
The above input distribution yields a symmetric-rate of $0.45330$
bits/transmission. Moreover, the $u^{*}$ corresponding to this
distribution is given by
\begin{align}
u^{*}&=\sum_{t}p_{t}(q_{1t}+q_{2t}-2q_{1t}q_{2t})\label{S23}\\
&=f(2u_{1}^{*},2u_{2}^{*})\label{S24}\\
&=0.355899\label{S25}
\end{align}
where (\ref{S24}) is by construction of the input distribution
$p(t,x_{1},x_{2})$ and (\ref{S25}) is obtained by substituting the
distribution specified in (\ref{S20})-(\ref{S22}). Moreover,
$\phi(2u_{2}^{*})< u^{*}<\frac{1}{2}$, hence we also have that
\begin{align}
\frac{1}{2}h(u^{*})&\geq \frac{1}{2}h(\phi(2u_{2}^{*}))
=h(\phi(2u_{1}^{*}))\label{S26}
\end{align}
This shows that the weakened version of the upper bound obtained in
(\ref{S8}) is indeed tight and a binary auxiliary random variable
$T$ with uniform distribution over $\{0,1\}$ is sufficient to attain
this symmetric-rate upper bound.

\section{Evaluation of the Cover-Leung Achievable Rate Region}
For completeness we will also obtain a simple characterization of
the Cover-Leung inner bound for our binary additive noisy MAC-FB.
For this purpose, we follow a two-step approach. In the first
step, we first obtain an outer bound on the achievable rate region
in terms of two variables $(u_{1},u_{2})$.  In the second step, we
specify an input distribution, as a function of $(u_{1},u_{2})$,
which achieves the outer bound. We therefore arrive at an
alternate characterization of the Cover-Leung achievable rate
region in terms of the variables $(u_{1},u_{2})$.

The Cover-Leung achievable rate region \cite{CL:1981} is given as,
\begin{align}
\mathcal{CL}=\Big\{(R_{1},R_{2}):\hspace{0.05in}&R_{1}\leq I(X_{1};Y|X_{2},T)\label{CL1}\\
&R_{2}\leq I(X_{2};Y|X_{1},T)\label{CL2}\\
&R_{1}+R_{2}\leq I(X_{1},X_{2};Y)\Big\}\label{CL3}
\end{align}
where the random variables $(T,X_{1},X_{2},Y)$ have the joint
distribution,
\begin{align}
p(t,x_{1},x_{2},y)&=p(t)p(x_{1}|t)p(x_{2}|t)p(y|x_{1},x_{2})\label{CL4}
\end{align}
and the random variable $T$ is subject to a cardinality constraint
of $|\mathcal{T}|\leq
\mbox{min}(|\mathcal{X}_{1}||\mathcal{X}_{2}|+1,\break
|\mathcal{Y}|+2)$. For the binary, additive noisy MAC in
consideration, the constraints in (\ref{CL1})-(\ref{CL3}) become,
\begin{align}
R_{1}&\leq \frac{1}{2}H(X_{1}|T)\label{CL6}\\
R_{2}&\leq \frac{1}{2}H(X_{2}|T)\label{CL7}\\
R_{1}+R_{2}&\leq I(X_{1},X_{2};Y)\label{CL8}
\end{align}
We will first obtain an outer bound on the region specified by
(\ref{CL6})-(\ref{CL8}) in terms of two variables $(u_{1},u_{2})$.
For every pair $(u_{1},u_{2})$, we will then specify an input
distribution which will attain this outer bound. Note that the
three constraints (\ref{CL6})-(\ref{CL8}) are of similar form as
in the case of $\mathcal{DB}_{PC}^{(1)}$ and
$\mathcal{DB}_{PC}^{(2)}$, and we proceed in a similar manner to
obtain upper bounds on the three terms above in terms of $u_{1}$
and $u_{2}$ as,
\begin{align}
R_{1}&\leq \frac{1}{2}h(\phi(2u_{1}))\label{CL9}\\
R_{2}&\leq \frac{1}{2}h(\phi(2u_{2}))\label{CL10}\\
R_{1}+R_{2}&\leq
h\left(\frac{1-f(2u_{1},u_{2})}{2}\right)\label{CL11}
\end{align}
where the variables $(u_{1},u_{2})$ belong to the set
$\mathcal{S}$ defined in (\ref{S15}). Hence, an outer bound on the
rate region specified by (\ref{CL6})-(\ref{CL8}) is given as
$\mathcal{O}$, where
\begin{align}
\mathcal{O}=\bigcup_{(u_{1},u_{2})\in
\mathcal{S}}\Bigg\{(R_{1},R_{2}): \hspace{0.05in}&R_{1}\leq \frac{1}{2}h(\phi(2u_{1}))\nonumber\\
&R_{2}\leq \frac{1}{2}h(\phi(2u_{2}))\nonumber\\
&R_{1}+R_{2}\leq
h\left(\frac{1-f(2u_{1},u_{2})}{2}\right)\label{CL14}\Bigg\}
\end{align}
Let $(u_{1},u_{2})$ be any arbitrary pair which belongs to
$\mathcal{S}$. Consider an input distribution for which
$|\mathcal{T}|=2$, and $T$ is uniform over $\{0,1\}$ and,
\begin{align}
p_{0}=p_{1}&=\frac{1}{2}\label{C16}\\
q_{10}=1-q_{11}&=\phi(2u_{1})\label{C17}\\
q_{20}=1-q_{21}&=\phi(2u_{2})\label{C18}
\end{align}
For this input distribution, we obtain the following exact
equalities
\begin{align}
H(X_{1}|T)&=h(\phi(2u_{1}))\label{C22}\\
H(X_{2}|T)&=h(\phi(2u_{2}))\label{C23}\\
I(X_{1},X_{2};Y)&=h\left(\frac{1-f(2u_{1},2u_{2})}{2}\right)\label{C24}
\end{align}
We have thus shown that the outer bound we obtained on the
achievable rate region in terms of $(u_{1},u_{2})$ can be attained
by a set of input distributions for which the involved auxiliary
random variable $T$ is binary and uniform. This in turn implies
that a binary and uniform random variable $T$ is sufficient to
characterize the entire Cover-Leung achievable rate region for the
binary additive noisy MAC-FB. By varying over all such input
distributions, or equivalently, by varying $(u_{1},u_{2})$ in the
set $\mathcal{S}$, we obtain the entire Cover-Leung achievable
rate region. We should remark here that when evaluating the
$\mathcal{DB}_{PC}$ bound in the previous section for $Z=X_{1}$
and $Z=X_{2}$, it was not necessary to specify the distribution
which achieves the bound, since it was an outer bound. On the
other hand, when evaluating the Cover-Leung bound, since it is an
achievability, it is necessary to give a distribution which
achieves the bound.

The dependence balance bounds corresponding to the parallel
channel choices $Z=X_{1}$ and $Z=X_{2}$, along with the cut-set
upper bound and the Cover-Leung achievable rate region are shown
in Figure $4$. It is interesting to note that our bound improves
upon the cut-set bound at all points where the Cover-Leung
achievable rate region is strictly larger than the capacity region
without feedback. In other words, our bound improves upon the
cut-set bound at all points where feedback increases capacity.

We should remark that our choices of parallel channels; namely,
$Z=X_{1}$ and $Z=X_{2}$ are the simplest ones which ensure that
$I(X_{1};X_{2}|Y,Z,T)=0$ but they yield fixed information leaks.
We believe that by a more elaborate choice of a parallel channel,
i.e., by carefully selecting a parameterized parallel channel
$p^{+}(z|x_{1},x_{2},y,t)$ such that $I(X_{1};X_{2}|Y,Z,T)=0$, one
would still be able to restrict the input distributions to a
conditionally independent form and then optimize the parameters of
the parallel channel to minimize the information leak terms. This
approach can potentially improve upon our outer bound.

\begin{figure}[t]
  \centering
  \centerline{\epsfig{figure=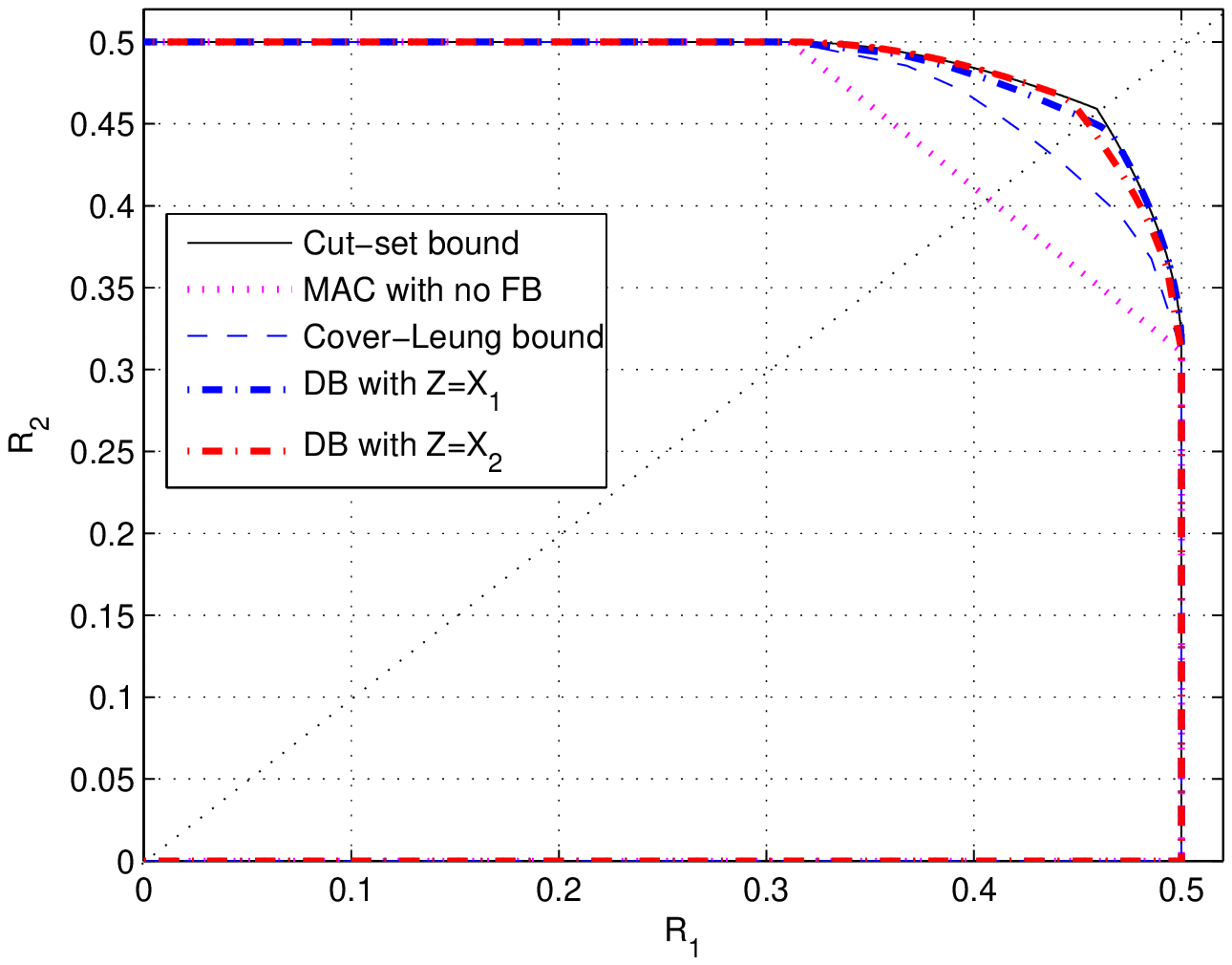,width=13cm}}
  { Figure $4.1$: Illustration of our bounds for the capacity of binary additive noisy MAC-FB.}
  \centering
  \centerline{\epsfig{figure=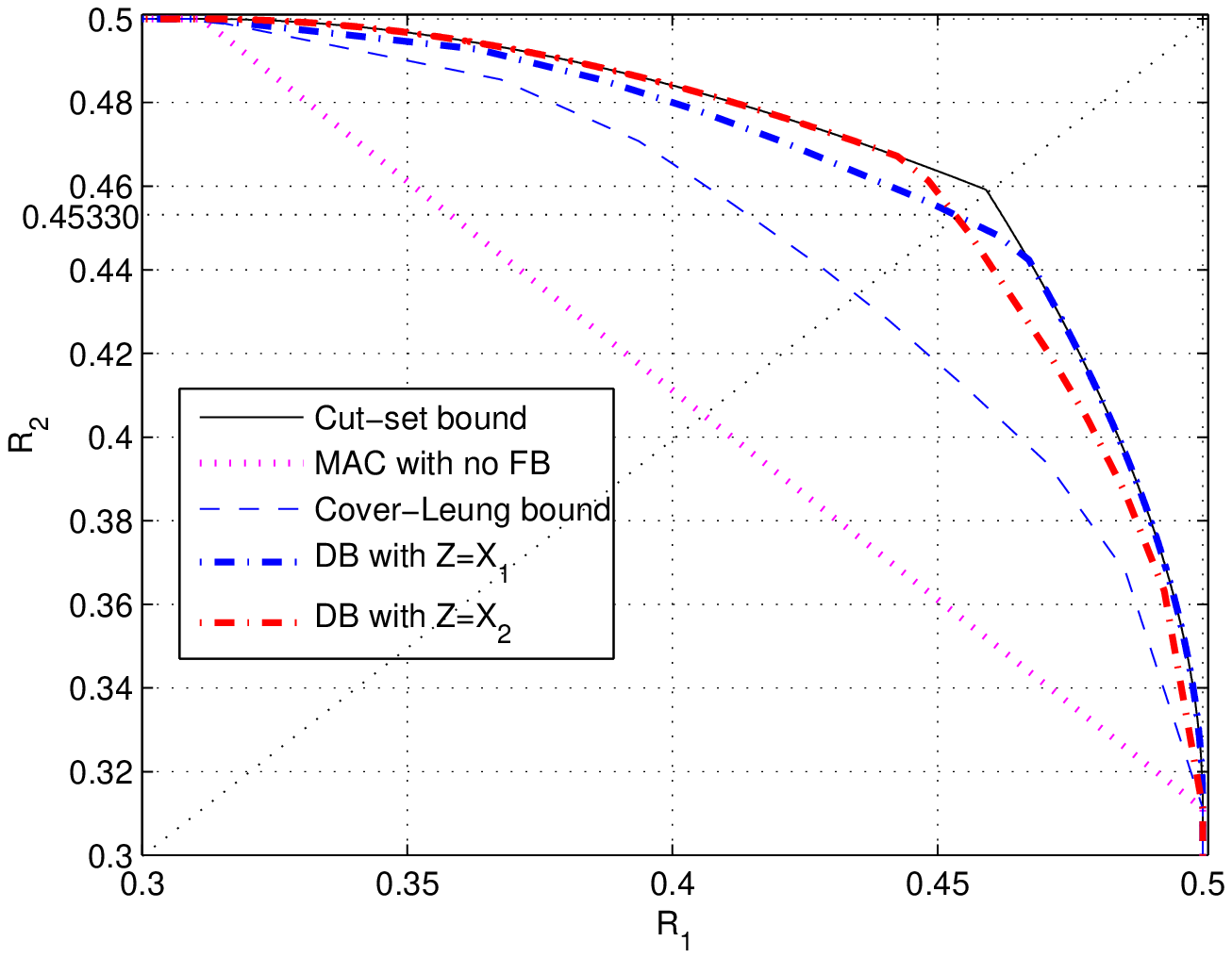,width=13cm}}
  { Figure $4.2$: An enlarged illustration of the portion of Figure $4.1$ where feedback increases capacity.}\medskip
\end{figure}

\section{The Capacity Region of the Binary Erasure MAC-FB}
The capacity region of a class of discrete memoryless MAC-FB was
characterized in \cite{WL:1982} by establishing a converse and it
was shown to be equal to the Cover-Leung achievable rate region.
This class of channels satisfy the property that at least one of
the channel inputs say $X_{1}$, can be written as a deterministic
function of the other channel input $X_{2}$ and the channel output
$Y$. The binary erasure MAC, where $Y=X_{1}+X_{2}$, falls into
this class of channels. In addition, the binary erasure MAC-FB is
the noiseless version of the binary additive noisy MAC-FB studied
in this paper.

Willems showed in \cite{WillemsonFB:1984} that a binary selection
of auxiliary random variable is sufficient to obtain the sum-rate
point of the capacity region of the binary erasure MAC-FB. In this
section, we will show that by using our results for composite
functions which were presented in previous sections, it is
possible to obtain all points on the boundary of this capacity
region using a binary auxiliary random variable. The feedback
capacity region of this channel is given by the Cover-Leung
achievable rate region given in (\ref{CL1})-(\ref{CL3}) which can
be simplified for the binary erasure MAC-FB as,
\begin{align}
R_{1}&\leq H(X_{1}|T)\label{E1}\\
R_{2}&\leq H(X_{2}|T)\label{E2}\\
R_{1}+R_{2}&\leq H(Y)\label{E3}
\end{align}
We obtain three upper bounds on the expressions appearing in the
bounds (\ref{E1})-(\ref{E3}). From (\ref{EV11}), we have,
\begin{align}
H(X_{1}|T)&\leq h(\phi(2u_{1}))\label{E7}
\end{align}
Similarly, we also have
\begin{align}
H(X_{2}|T)&\leq h(\phi(2u_{2}))\label{E8}
\end{align}
We now obtain an upper bound on $H(Y)$, by first noting that,
\begin{align}
H(Y)&=h^{(3)}(P_{Y}(0),P_{Y}(1),P_{Y}(2))\label{E9}
\end{align}
where
\begin{align}
P_{Y}(0)&=\sum_{t}p_{t}q_{1t}q_{2t}\label{E10}\\
P_{Y}(1)&=\sum_{t}p_{t}(q_{1t}+q_{2t}-2q_{1t}q_{2t})\label{E11}\\
P_{Y}(2)&=\sum_{t}p_{t}(1-q_{1t})(1-q_{2t})\label{E12}
\end{align}
Now, we use the following inequality established in
\cite{WillemsonFB:1984},
\begin{align}
h^{(3)}(a,b,c)&=\frac{1}{2}h^{(3)}(a,b,c)+\frac{1}{2}h^{(3)}(c,b,a)\label{E13}\\
&\leq h^{(3)}\left(\frac{a+c}{2},b,\frac{a+c}{2}\right)\label{E14}\\
&=h(b)+1-b\label{E15}
\end{align}
where (\ref{E14}) follows by the concavity of the entropy function
and by the application of Jensen's inequality \cite{Cover:book}.
Using (\ref{E15}) and continuing from (\ref{E9}), we obtain
\begin{align}
H(Y)&=h^{(3)}(P_{Y}(0),P_{Y}(1),P_{Y}(2))\label{E16}\\
&\leq h(P_{Y}(1))+1-P_{Y}(1)\label{E17}\\
&= h(u)+1-u\label{E18}
\end{align}
where $u$ is defined in (\ref{EV4}). Using (\ref{E7}), (\ref{E8})
and (\ref{E18}), we can write an outer bound $\mathcal{O}_{1}$ on
the capacity region as follows,
\begin{align}
\mathcal{O}_{1}&=\bigcup_{(u_{1},u_{2},u)\in
\mathcal{P}}\mathcal{O}_{1}(u_{1},u_{2},u)\label{E28}
\end{align}
where
\begin{align}
\mathcal{O}_{1}(u_{1},u_{2},u)=\Big\{(R_{1},R_{2}):
\hspace{0.05in}&R_{1}\leq
h(\phi(2u_{1}))\nonumber\\
& R_{2}\leq h(\phi(2u_{2}))\nonumber\\
& R_{1}+R_{2} \leq h(u)+1-u \label{E21}\Big\}
\end{align}
and the set $\mathcal{P}$ is defined in (\ref{P23}). We will now
obtain a simpler characterization of $\mathcal{O}_{1}$ in terms of
two variables $(u_{1},u_{2})$ by showing that
$\mathcal{O}_{1}\equiv\mathcal{O}_{2}$, where,
\begin{align}
\mathcal{O}_{2}&=\bigcup_{(u_{1},u_{2})\in
\mathcal{S}}\mathcal{O}_{2}(u_{1},u_{2})\label{E33}
\end{align}
where
\begin{align}
\mathcal{O}_{2}(u_{1},u_{2})=\Big\{(R_{1},R_{2}):
\hspace{0.05in}&R_{1}\leq
h(\phi(2u_{1}))\nonumber\\
& R_{2}\leq h(\phi(2u_{2}))\nonumber\\
& R_{1}+R_{2} \leq h(f(2u_{1},2u_{2}))+1-f(2u_{1},2u_{2})
\label{E32}\Big\}
\end{align}
The inclusion $\mathcal{O}_{2}\subseteq \mathcal{O}_{1}$ is
straightforward by forcing $u=f(2u_{1},2u_{2})$ in
$\mathcal{O}_{1}$. We will now show that $\mathcal{O}_{1}\subseteq
\mathcal{O}_{2}$. For this purpose, we will need the following
lemma.
\begin{Lem} The function
\begin{align}
\mu(s)&=h(s)+1-s\label{E34}
\end{align}
is concave in $s$ for $s\in [0,1]$ and takes its maximum value at
$s=\frac{1}{3}$. Moreover, the function $\mu(s)$ is increasing in
$s$ for $s\in [0,\frac{1}{3}]$ and decreasing in $s$ for $s\in
[\frac{1}{3},1]$.
\end{Lem}
The proof of this lemma follows from the fact that both $h(s)$ and
$-s$ are concave in $s$.

Now consider any arbitrary triple $(u_{1},u_{2},u)\in \mathcal{P}$.
We can classify any such triple into one of the following cases:
\begin{enumerate}
\item If $f(2u_{1},2u_{2})\leq u\leq \frac{1}{2}$: for any such
$(u_{1},u_{2},u)$, there exists a pair
$(\bar{u}_{1},\bar{u}_{2})$, such that
\begin{align}
u_{1}&\leq \bar{u}_{1}\leq \frac{1}{4}\label{E35}\\
u_{2}&\leq \bar{u}_{2}\leq \frac{1}{4}\label{E36}\\
u&=f(2\bar{u}_{1},2\bar{u}_{2})\label{E37}
\end{align}
One such pair $(\bar{u}_{1},\bar{u}_{2})$ can be obtained as
follows. Using the fact that for a fixed $u_{1}$,
$f(2u_{1},2u_{2})$ is increasing in $u_{2}$, we select
$\bar{u}_{1}=u_{1}$ and solve for $u_{2}\leq \bar{u}_{2} \leq
\frac{1}{4}$ for which $f(2\bar{u}_{1},2\bar{u}_{2})=u$. The
required $\bar{u}_{2}$ is obtained as,
\begin{align}
\bar{u}_{2}&=\frac{1}{4}\left(1-\frac{(1-2u)^{2}}{(1-4u_{1})}\right)\label{E39}
\end{align}
For such a pair $(\bar{u}_{1},\bar{u}_{2})$, the following
inequalities hold,
\begin{align}
h(\phi(2u_{1}))&= h(\phi(2\bar{u}_{1}))\label{E40}\\
h(\phi(2u_{2}))&\leq h(\phi(2\bar{u}_{2}))\label{E41}\\
h(u)+1-u&=
h(f(2\bar{u}_{1},2\bar{u}_{2}))+1-f(2\bar{u}_{1},2\bar{u}_{2})\label{E42}
\end{align}
\item If $f(2u_{1},2u_{2})\leq \frac{1}{2} \leq u \leq
1-(u_{1}+u_{2})$, then we have by Lemma $4$,
\begin{align}
h(u)+1-u &\leq h\left(\frac{1}{2}\right)+1-\frac{1}{2}\label{E43}\\
&=\frac{3}{2} \label{E44}
\end{align}
Now consider the pair
$(\bar{u}_{1},\bar{u}_{2})=(\frac{1}{4},\frac{1}{4})$, for which
we have $f(2\bar{u}_{1},2\bar{u}_{2})=\frac{1}{2}$. Hence we have
that,
\begin{align}
h(\phi(2u_{1}))&\leq h(\phi(2\bar{u}_{1}))=1\label{E45}\\
h(\phi(2u_{2}))&\leq h(\phi(2\bar{u}_{2}))=1\label{E46}\\
h(u)+1-u&\leq
h(f(2\bar{u}_{1},2\bar{u}_{2}))+1-f(2\bar{u}_{1},2\bar{u}_{2})=\frac{3}{2}\label{E47}
\end{align}
\end{enumerate}
We have thus shown that for any triple $(u_{1},u_{2},u)$, there
exists a pair $(\bar{u}_{1},\bar{u}_{2})$, such that
$\mathcal{O}_{1}(u_{1},u_{2},u)\subseteq
\mathcal{O}_{2}(\bar{u}_{1},\bar{u}_{2})$, which in turn implies
that $\mathcal{O}_{1}\subseteq \mathcal{O}_{2}$, and consequently
$\mathcal{O}_{1}\equiv \mathcal{O}_{2}$. Hence, we have an outer
bound on the capacity region as given by $\mathcal{O}_{2}$.

The outer bound $\mathcal{O}_{2}$ is evaluated over the set of
pairs $(u_{1},u_{2})$ such that $u_{1},u_{2}\in [0,\frac{1}{4}]$.
For any such arbitrary pair $(u_{1},u_{2})$,  an input
distribution which achieves the set of rate pairs specified by
$\mathcal{O}_{2}(u_{1},u_{2})$ is obtained by selecting
$|\mathcal{T}|=2$, and
\begin{align}
p_{0}=p_{1}&=\frac{1}{2}\label{E48}\\
q_{10}=1-q_{11}&=\phi(2u_{1})\label{E49}\\
q_{20}=1-q_{21}&=\phi(2u_{2})\label{E50}
\end{align}
The set of rates achievable by the distribution specified in
(\ref{E48})-(\ref{E50}) are obtained as,
\begin{align}
R_{1}&\leq H(X_{1}|T) =h(\phi(2u_{1}))\label{E51}\\
R_{2}&\leq H(X_{2}|T)=h(\phi(2u_{2}))\label{E52}\\
R_{1}+R_{2}&\leq
H(Y)=h(f(2u_{1},2u_{2}))+1-f(2u_{1},2u_{2})\label{E53}
\end{align}

This shows that the capacity region of binary erasure MAC-FB can
be obtained by a binary and uniform selection of the auxiliary
random variable $T$. The capacity region of the binary erasure MAC
with and without feedback and the cut-set bound are illustrated in
Figure $5$. It was shown in \cite{WillemsonFB:1984} that the
sum-rate point on the boundary of the capacity region lies
strictly below the ``total cooperation" line. This is equivalent
to saying that the cut-set bound is not tight for the sum-rate
point. From our result, it is now clear that the cut-set bound is
not tight for asymmetric rate pairs either. In fact, it is not
tight at all boundary points where feedback increases capacity.

Moreover, our result also shows that a simple selection of binary
and uniform $T$ is sufficient to evaluate the boundary of the
capacity region of binary erasure MAC-FB. Simple feedback
strategies for a class of two user MAC-FB were developed in
\cite{KramerMAC:1999}. It was shown that for the binary erasure
MAC, these feedback strategies yield all rate points for a binary
selection of the auxiliary random variable $T$. Thus, our result
shows that these feedback strategies are indeed optimal for the
binary erasure MAC-FB and yield all rates on the boundary of its
feedback capacity region.

\section{Conclusions}
In this paper, we obtained a new outer bound on the capacity
region of a MAC-FB by using the idea of dependence balance. We
considered a binary additive noisy MAC-FB for which it is known
that feedback increases capacity but the feedback capacity region
is not known. The best known outer bound on the feedback capacity
region of this channel was the cut-set bound. We used the
dependence balance bound to improve upon the cut-set bound at all
points in the capacity region of this channel where feedback
increases capacity. Our result is somewhat surprising once it is
realized that the channel we considered in this paper is the
discrete version of the two-user Gaussian MAC-FB considered by
Ozarow in \cite{OZ:1984} where the cut-set bound was shown to be
tight.

Our outer bound is difficult to evaluate due to an involved
auxiliary random variable $T$. For binary inputs, the cardinality
bound on $T$ is $|\mathcal{T}|\leq 7$ which makes it intractable
to evaluate the outer bound. We overcome this difficulty by making
use of composite functions and their properties to obtain a simple
characterization of our bound. As an application of the properties
of the composite functions developed in this paper, we are also
able to completely characterize the Cover-Leung achievable rate
region for this channel.

The capacity region of the binary erasure MAC-FB is known and it
coincides with the Cover-Leung achievable rate region. Although the
capacity region is known in principle, it is not known how to
compute the entire region, the difficulty arising again due to the
involved auxiliary random variable. We again make use of the
composite functions to give an alternate characterization of the
capacity region of the binary erasure MAC-FB. In addition, we go on
to show that a binary and uniform auxiliary random variable
selection is sufficient to evaluate its feedback capacity region.

\begin{figure}[t]
  \centering
  \centerline{\epsfig{figure=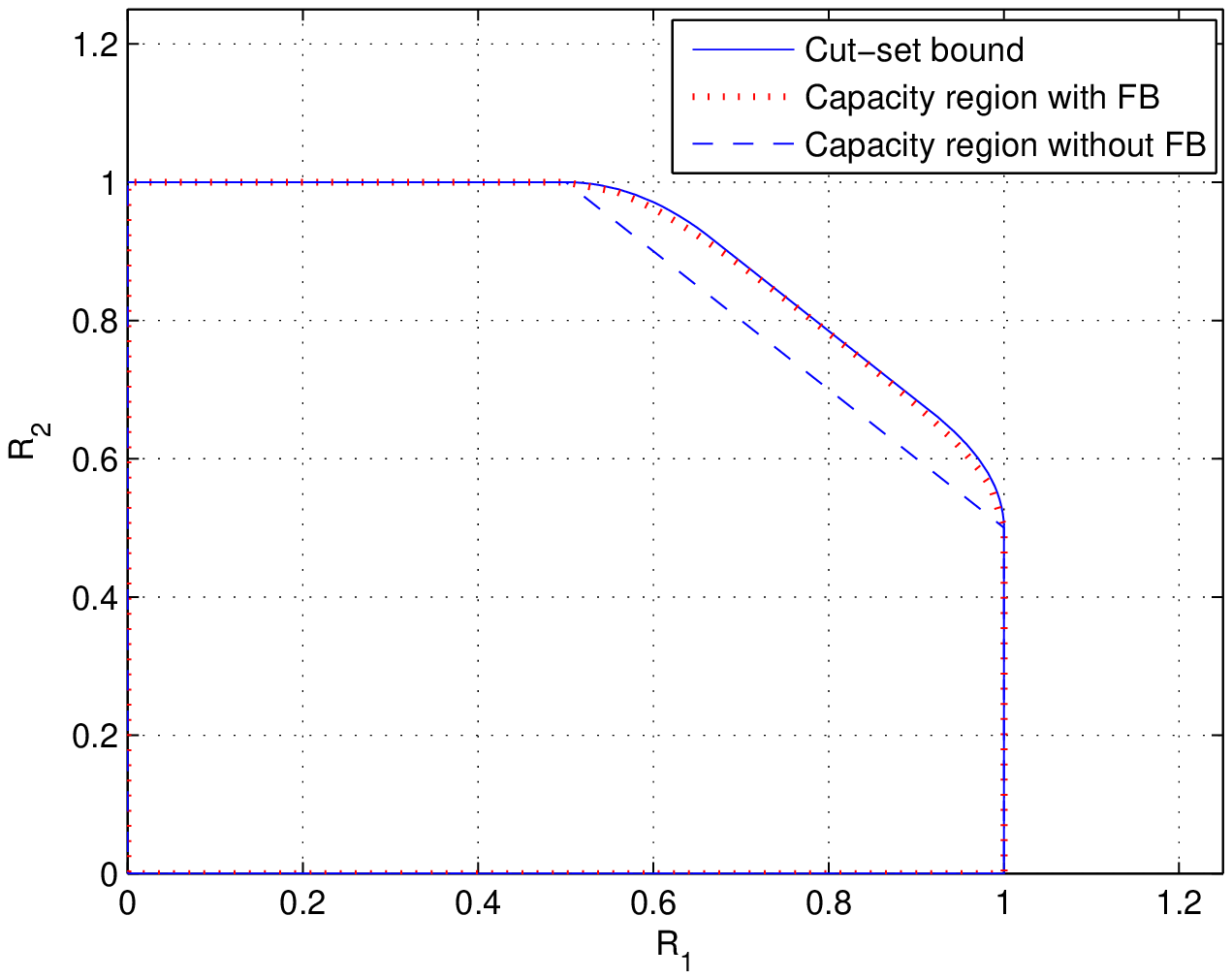,width=12cm}}
  { Figure $5.1$: Illustration of the capacity region of binary erasure MAC-FB.}
  \centering
  \centerline{\epsfig{figure=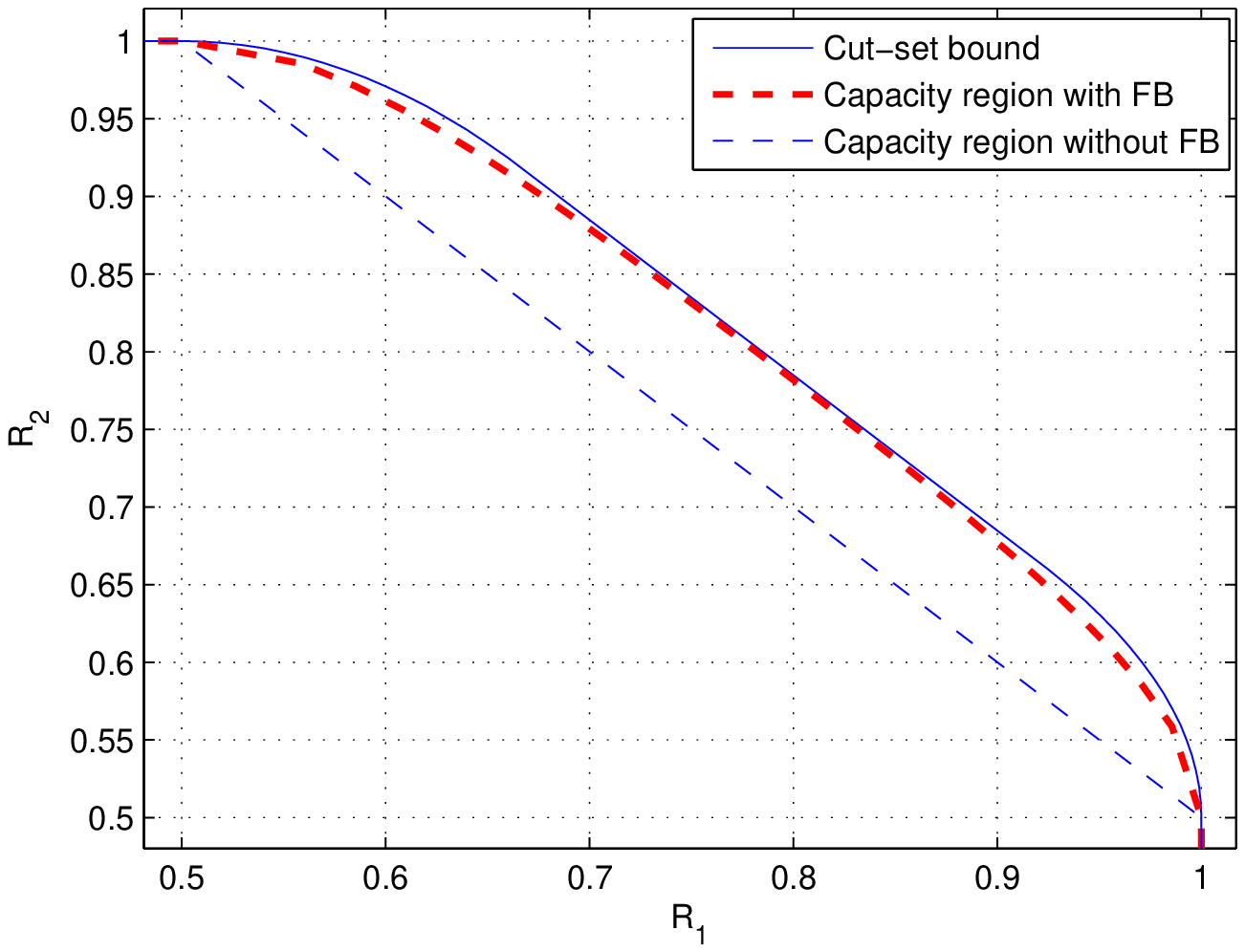,width=12cm}}
  { Figure $5.2$: An enlarged illustration of the portion of Figure $5.1$ where feedback increases capacity.
 }\medskip
\end{figure}

\bibliographystyle{unsrt}
\bibliography{journal}

\begin{thebibliography}{10}

\bibitem{Hekstra_Willems:1989}
A.~P. Hekstra and F.~M.~J. Willems.
\newblock Dependence balance bounds for single output two-way channels.
\newblock {\em IEEE Trans. on Information Theory}, 35(1):44--53, January 1989.

\bibitem{OZ:1984}
L.~Ozarow.
\newblock The capacity of the white {G}aussian multiple access channel with
  feedback.
\newblock {\em IEEE Trans. on Information Theory}, 30(4):623--629, July 1984.

\bibitem{CL:1981}
T.~M. Cover and C.~S.~K. Leung.
\newblock An achievable rate region for the multiple access channel with
  feedback.
\newblock {\em IEEE Trans. on Information Theory}, 27(3):292--298, May 1981.

\bibitem{KramerMAC:1999}
G.~Kramer.
\newblock Feedback strategies for a class of two-user multiple access channels
  with feedback.
\newblock {\em IEEE Trans. on Information Theory}, 45(6):2054--2059, September
  1999.

\bibitem{GW:1975}
N.~Gaarder and J.~Wolf.
\newblock The capacity region of a multiple-access discrete memoryless channel
  can increase with feedback.
\newblock {\em IEEE Trans. on Information Theory}, 21:100--102, Jan 1975.

\bibitem{SK:1966}
J.~P.~M. Schalkwijk and T.~Kailath.
\newblock A coding scheme for additive noise channels with feedback-{P}art {I}:
  No bandwidth constraint.
\newblock {\em IEEE Trans. on Information Theory}, 12:172--182, April 1966.

\bibitem{Kramer:thesis}
G.~Kramer.
\newblock {\em Directed Information for Channels with Feedback}.
\newblock Ph.D. dissertation, Swiss Federal Institute of Technology (ETH),
  Zurich, Switzerland, 1998.

\bibitem{BrossLapidoth:2005}
S.~I. Bross and A.~Lapidoth.
\newblock An improved achievable rate region for the discrete memoryless
  two-user multiple-access channel with noiseless feedback.
\newblock {\em IEEE Trans. on Information Theory}, 51(3):811--833, March 2005.

\bibitem{WL:1982}
F.~M.~J. Willems.
\newblock The feedback capacity region of a class of discrete memoryless
  multiple access channels.
\newblock {\em IEEE Trans. on Information Theory}, 28(1):93--95, January 1982.

\bibitem{Kramer:2003}
G.~Kramer.
\newblock Capacity results for the discrete memoryless network.
\newblock {\em IEEE Trans. on Information Theory}, 49(1):4--21, Jan. 2003.

\bibitem{VinckPost:1985}
A.~J. Vinck, W.~L.~M. Hoeks, and K.~A. Post.
\newblock On the capacity of the two-user {M}-ary multiple-access channel with
  feedback.
\newblock {\em IEEE Trans. on Information Theory}, 31(4):540--543, July 1985.

\bibitem{WillemsonFB:1984}
F.~Willems.
\newblock On multiple access channels with feedback.
\newblock {\em IEEE Trans. on Information Theory}, 30(6):842--845, November
  1984.

\bibitem{Cover:book}
T.~M. Cover and J.~A. Thomas.
\newblock {\em Elements of Information Theory}.
\newblock New York:Wiley, 1991.

\bibitem{Boyd:book}
S.~Boyd and L.~Vandenberghe.
\newblock {\em Convex Optimization}.
\newblock Cambridge University Press, 2004.

\end{thebibliography}

\end{document}